\documentclass[12pt]{article}
\usepackage[margin=1.25in]{geometry}
\usepackage{longtable}
\usepackage{lscape}
\usepackage{graphicx}
\usepackage{csvsimple}
\usepackage{siunitx}
\usepackage{microtype}
\usepackage{booktabs}

\usepackage{setspace}
\usepackage{tablefootnote}
\usepackage[table]{xcolor}
\newcommand{\E}{\mathbb{E}}

\usepackage{amssymb}
\usepackage{amsmath}
\usepackage{amsthm}
\usepackage{amsfonts}
\usepackage{mathrsfs}
\usepackage{appendix}
\usepackage[utf8]{inputenc}
\usepackage[authoryear]{natbib}
\usepackage{float}
\usepackage{graphicx}
\usepackage{caption}
\usepackage{subcaption}
\usepackage{booktabs}
\usepackage{comment}
\usepackage{xcolor}
\usepackage{tabularx}
\usepackage{tikz}
\usepackage{bbm}
\usepackage{adjustbox}
\usepackage{makecell}
\usepackage{diagbox}
\usepackage[hidelinks]{hyperref}

\usepackage{multicol}
\usepackage{multirow}
\theoremstyle{plain}
\newtheorem{definition}{Definition}
\newtheorem{prop}{Proposition}
\newtheorem{thm}{Theorem}
\newtheorem{lemma}{Lemma}

\newtheorem{corr}{Corollary}
\newtheorem{assmp}{Assumption}

\makeatletter
\providecommand{\@LN}[2]{}
\makeatother

\DeclareMathOperator{\Cov}{Cov}
\DeclareMathOperator{\Var}{Var}
\DeclareMathOperator{\sign}{Sign}

\DeclareMathOperator{\Max}{Max}

\DeclareMathOperator{\Min}{Min}

\begin{document}

\title{{\huge{}Monotone Ecological Inference}\thanks{For helpful comments and suggestions, we are grateful to Kosuke Imai, Gary King, Charles Manski, Shiying Hao, and Derek Ouyang. Elzayn, Ho, and Morton: Stanford University. Goldin: University of Chicago, American Bar Foundation, and NBER. Guage: Columbia University.}}

\author{Hadi Elzayn \and Jacob Goldin \and Cameron Guage \and Daniel E. Ho \and Claire Morton}

\date{March, 2025}

\clearpage
\maketitle
\thispagestyle{empty}

\begin{abstract}
\noindent We study monotone ecological inference, a partial identification approach to ecological inference. The approach exploits information about one or both of the following conditional associations: (1) outcome differences between groups within the same neighborhood, and (2) outcomes differences within the same group across neighborhoods with different group compositions. We show how assumptions about the sign of these conditional associations, whether individually or in relation to one another, can yield informative sharp bounds in ecological inference settings. We illustrate our proposed approach using county-level data to study differences in Covid-19 vaccination rates among Republicans and Democrats in the United States.
\end{abstract}

\doublespacing
\clearpage
\setcounter{page}{1}

\section{Introduction}

Ecological inference (EI) -- the use of aggregate data to investigate individual-level associations -- is a common challenge in fields such as political science, sociology, economics, epidemiology, and public health. To fix ideas, consider a researcher seeking to learn the difference in the prevalence of some outcome across two groups of individuals, but data on both the outcome and group membership are available only at some aggregated level, such as the individual’s neighborhood. EI is challenging because the neighborhood-level data can be consistent with divergent individual-level associations between group membership and the outcome. 

There are several canonical estimation strategies that have been proposed for EI settings, but each requires strong assumptions. Ecological regression \citep{goodman1953ecological} entails regressing neighborhood-level outcomes on neighborhood-level rates of group membership. The approach is unbiased only absent ``contextual effects'' -- e.g., it would be biased if one group is more likely to live in neighborhoods where individuals of both groups tend to have higher values of the outcome. Alternatively, the neighborhood model \citep{freedman1991ecological} uses group-weighted averages across neighborhoods to estimate group means. It yields unbiased  estimates under a different, but equally strong assumption, namely that there is no association between group membership and the outcome among individuals living in the same neighborhood. A third approach to EI, proposed by \citet{king1997solution}, relies for identification on an assumption about the role of contextual effects within a parametric statistical model.\footnote{Specifically, King models the individual-level variables of interest as independent draws from a truncated bivariate normal distribution, where truncation ensures that the estimated proportions fall in the unit interval. It provided the first statistical approach to combine the method of bounds with ecological regression.}

A limitation to all of these approaches is that they require strong assumptions for point identification. A potentially appealing alternative is therefore to focus on partial identification methods, such as the so-called method of bounds \citep{Duncan1953AnAT}, which yields the range of individual-level associations potentially consistent with the observed aggregated data. In practice, however, the method of bounds interval is frequently wide. Depending on the application, such bounds may not shed much light on the parameter of interest.

In this paper, we propose a middle-ground approach to identification in EI settings that strikes a balance between obtaining informative results while relying on potentially more credible identifying assumptions. The identifying assumptions we consider concern two quantities: (1) the between-group within-neighborhood association -- i.e, outcome differences between groups within the same neighborhood, and (2) the within-group between-neighborhood association -- i.e., outcome differences within the same group across neighborhoods with different group compositions. We show how assumptions about the \emph{sign} of these conditional associations -- whether individually or in relation to one another -- can aid identification. In particular, we derive sharp bounds for the group-level outcome means and for the difference in outcome means by group in EI settings in which the researcher can sign one or both of the conditional associations. Because our identifying assumptions relate to the signs of the conditional associations, we refer to our proposed method as \emph{monotone ecological inference}, in the spirit of \citet{manski_pepper_2000}.

A virtue of monotone ecological inference is that it relies on weaker identifying assumptions than canonical EI estimators, in that it does not require taking a stance on the exact magnitude of either the between- or the within-group association. In contrast, the neighborhood model is unbiased only if the between-group association is exactly zero, and ecological regression is unbiased only if the within-group association is exactly zero.\footnote{King's approach also imposes that the within-group association be zero a priori, but allows for a non-zero conditional association after incorporating feasibility constraints based on the neighborhood-level data \citep[for discussion, see][]{lewis2001understanding,jiang2020ecological}.} These distinctions are important in practice because in many settings of interest, the researcher will not be able to entirely rule out within-neighborhood variation in group outcomes (so that the between-group association may be non-zero) nor be able to rule out the presence of all contextual effects (so that the within-group association may be non-zero). At the same time, the researcher may be able to use theory or auxiliary data to form a belief about the likely direction of any such associations if they are present, enabling monotone ecological infernece.\footnote{As we illustrate below, an auxiliary sample of individual-level data may be useful to assess the sign of one or both of the conditional associations, even if the sample is not sufficiently large or representative enough to answer the primary research question on its own.}

A related possibility is that the researcher may have reason to believe that the two conditional associations run in the same direction as one another, even without necessarily knowing the direction. For example, the researcher may expect that the contextual effects in a particular setting would amplify any group differences in the outcome that would otherwise exist, such as through social norms or local control over policy-making, so that the within-group (between-neighborhood) association would tend to have the same sign as the between-group (within-neighborhood) association. Under this condition, which we refer to as \emph{contextual reinforcement}, we show that the data identifies the sign of the difference in group means, and we provide sharp bounds for the group means as well as for their difference. 

The main reason that monotone ecological inference is appealing as a research design is that its identifying assumptions can often be reasoned about on the basis of expert institutional knowledge and/or auxiliary data. We illustrate this type of reasoning in our empirical application, where we study differences in Covid-19 vaccination rates by political party using county-level data. This topic has been the focus of substantial interest in recent years, and much of the prior evidence is ecological in nature \citep{albrecht2022vaccination,ye2023exploring}.

In applying monotone ecological inference to this research question, we start by noting that the setting is one in which contextual reinforcement is likely to hold. In particular, because many public health services and policies like vaccination mandates are determined and applied on a sub-national level, we expect that any within-neighborhood association between political party and vaccination status would be amplified through neighborhood-level mediators. For example, suppose that Republicans are more vaccine hesitant for reasons that are not fully mediated through neighborhood (e.g., differences in partisan media exposure and attitude formation), so that the between-group association is negative. We would then also expect that Republicans living in “red” counties (where most of the residents are Republicans) would be more likely to face local policies and social norms inhibiting vaccinations, as compared to Republicans living in “blue” counties (where few of the residents are Republicans); hence the within-group (across-county) association would also be negative.\footnote{Conversely, if Democrats tended to be more opposed to vaccination, we would expect the opposite pattern to obtain for both identifying assumptions.} The plausibility of such ``feedback’’ from individual- to neighborhood-level associations supports the assumption of contextual reinforcement in this setting. 

To augment the theoretical support for contextual reinforcement, we obtain an auxiliary data set consisting of individual-level vaccination records for approximately one million individuals from a registry of electronic health records. We link these records to publicly available data on voter registration and then anonymize them. For this sample, we find that both the between- and within-group associations are negative: Republicans are less likely to be vaccinated than Democrats living in the same county, and individuals of both parties are more likely to be vaccinated in counties where individuals of their same party have higher vaccination rates. Although the exact magnitude of these associations is unlikely to be the same in the overall population, we interpret the sign of the estimated associations for this subpopulation to support the contextual reinforcement assumption. 

Analyzing the county-level data using monotone ecological inference, we find evidence that Democrats are vaccinated at higher rates than Republicans. We estimate the vaccination rate among Democrats to be between 0.61 and 0.81, compared to between 0.33 and 0.56 among Republicans. With respect to the difference in vaccination rates between Democrats and Republicans, the width of the contextual reinforcement bounds is approximately 67\% narrower than the method of bounds interval. And unlike the method of bounds, the contextual reinforcement bounds exclude the possibility that there is no difference in vaccination rates by political party.

Our results contribute to a large literature that studies identification in EI settings;  \citep[for surveys see, e.g.,][]{king1997solution, king2004ecological, cho200824}. Most of this literature focuses on point-identification, with notable exceptions that include \citet{Duncan1953AnAT}, \citet{horowitz1995identification}, \citet{cross2002regressions}, \citet{greiner}, \citet{manski2018credible}, and \citet{jiang2020ecological}. Among these,  \citet{manski2018credible} in particular shares a key feature of our approach, which is to consider sign restrictions on the joint distribution of individual-level variables to aid in identification. However, as we detail below, our primary focus is on identifying a different parameter, which leads us to a very different set of results and insights. More recently, \citet{li2023partial}  also consider the role of ``bounded variation assumptions'' for identifying personalized risk assessments from published medical studies, a setting that shares some features of ecological inference but also differs from it in important respects. We contribute to this literature by proposing a novel partial identification strategy that can yield informative bounds under appealing identifying assumptions. We also provide novel expressions for the biases associated with canonical EI estimators, clarifying the interpretation of results when they are applied.

Outside of the EI setting, our results relate to a literature that infers disparities in individual-level data based on probabilistic estimates of group membership \citep{chen2019fairness,kallus2022assessing, mccartan2024estimating}. Closest to our approach is \citet{elzayn2024measuring}, which applies a partial identification strategy for estimating income tax audit disparities by race using individual-level data and probabilistically inferred racial characteristics based on the sign of conditional covariance terms that are the individual-level analogs to the conditional associations we study.\footnote{\citet{elzayn2024measuring} does not relate its approach to ecological inference nor derive the sharp bounds we provide here.} Our approach is also similar in spirit to prior work that studies the identification of treatment effects when the researcher substitutes identifying assumptions based on inequalities for identifying assumptions based on equalities \citep[e.g.,][]{manski_pepper_2000, molinari2010missing}.

We proceed as follows. Section 2 describes our empirical setup and defines several canonical EI estimators. Section 3 derives sharp bounds for the population difference in group means using Monotone EI. Section 4 extends the same approach to identification of the levels of group means. Section 5 develops Monotone EI tools for identification of group means at the neighborhood-level. Section 6 applies our results to study partisan gaps in Covid-19 vaccinations. Section 7 concludes. An open-source R software package, MonotoneEI, is available to implement our proposed approach. The software and accompanying
documentation are available at \url{https://github.com/reglab/MonotoneEI}.

\section{Empirical Framework}

A population of individuals is characterized by a triple $(X,Y,N)$. We interpret $X$ to define the individual's group membership; $Y$ to be the outcome of interest; and $N$ as the individual's neighborhood. For expositional simplicity, we initially focus on the case in which $X\in\{0,1\}$; the analysis extends naturally to settings in which there are more than two discrete groups, as we discuss below. We are primarily interested in the mean values of $Y$ by group, $$Y^x=\mathbb{E}[Y|X=x],$$ as well as the difference in group means, $$D=Y^1-Y^0.$$

Following the presentation of our main results, we also consider the identification of the levels and difference of the group means for specific neighborhoods, $\mathbb{E}[Y |X,N=n]$ and $\mathbb{E}[Y | X=1,\,N=n] - \mathbb{E}[Y | X=0,\,N=n]$.

We do not observe individual-level values of $X$ or $Y$, but rather observe data that has been aggregated across individuals with the same value of $N$. Denote the mean values of $X$ and $Y$ for the individuals in a particular neighborhood $n$ by
$$X_n=\mathbb{E}[X|N=n]$$ and $$Y_n=\mathbb{E}[Y|N=n].$$

The distribution of individuals across neighborhoods is given by  $$p_n=\text{Pr}(N=n).$$ 

We assume the researcher can directly observe $(X_n,Y_n,p_n)$ for each $n$, deferring issues of sampling uncertainty and statistical inference to our empirical application.  We also assume that there are a finite number of neighborhoods, $N\in\mathcal{N}$ with $|\mathcal{N}|<\infty$. To avoid degenerate cases, we assume that there exists some $n\in\mathcal{N}$ for which $p_n>0$ and $X_n\in(0,1)$. In addition, we assume that for each group $x$ and neighborhood $n$, the conditional expectation of $Y$ with respect to $x$ and $n$ exists and is bounded, $\underline{Y}\leq Y_n^x\leq \overline{Y}$ for some pair $(\underline{Y},\overline{Y})\in \mathbb{R}^2$, where  $Y_n^x=\mathbb{E}[Y|N=n, X=x]$.

Figure 1 visualizes the joint relationship between $X$, $Y$, and $N$. The overall population association between $X$ and $Y$ can be decomposed into two conditional associations. We formally define the \emph{between-group association} as
\begin{equation*}
    \delta_B := \mathbb{E}\left[\text{Cov} \left( Y , X \, | \,N \right) \right]
\end{equation*}
and the \emph{within-group association} as:
\begin{equation*}
    \delta_W := \mathbb{E}\left[\text{Cov} \left( Y,X_N \, | \,X \right) \right]
\end{equation*}
where $X_N=\mathbb{E}[X|N]$. Intuitively, the between-group association, $\delta_B$, refers to differences in the outcome between groups within the same neighborhood. In turn, the within-group association, $\delta_W$, refers to differences in the outcome across neighborhoods with different group prevalence, among individuals within the same group.\footnote{These interpretation follows from the fact that $X$ is binary. In particular, $\Cov(X,Y|N)=\E[XY|N]-\E[Y|N]\E[X|N]=\Pr[X=1|N](\E[Y|X=1,N]-\E[Y|N])=\Var[X|N](Y_N^1-Y_N^0)$. Hence,
$\delta_B = \sum_{n\in\mathcal{N}}(Y_n^1-Y_n^0)p_n\Var[X|n]$. Similarly, $\delta_W=\mathbb{E}[X]\,\text{Cov}\left(Y^1_n,X_n\right) + (1-\mathbb{E}[X])\,\text{Cov}\left(Y^0_n,X_n\right)$.} 

It is possible to relate the overall difference in group means, $D$, to these conditional associations, $\delta_B$ and $\delta_W$. To do so, define $\gamma=\frac{\text{Var}(X_N)}{\text{Var}(X)}$. Intuitively, $\gamma$ reflects the loss of information due to the data being aggregated, so that $\gamma<1$ in EI settings. We then have the following result. 

\begin{prop}[Decomposition of Overall Association into Conditional Associations]
\label{prop:decomp_conditional_assoc}
\begin{equation*}
    D = \frac{\delta_B + \delta_W}{(1-\gamma)\Var(X)}
\end{equation*}
\end{prop}
\noindent The proofs of Proposition \ref{prop:decomp_conditional_assoc} and all subsequent results are provided in the Appendix. Because $\gamma<1$, an immediate implication of Proposition \ref{prop:decomp_conditional_assoc} is that the sign of $D$ is the same as that of $\delta_B+\delta_W$.

\subsection*{Canonical EI Estimators}

We next describe three canonical approaches for estimating the group-level means in EI settings. We focus on these three methods (as opposed to other popular methods like \citet{king1997solution}) because they will appear as inputs into the monotone EI bounds we derive.

The most common method for conducting EI is 
ecological regression (ER) \citep{goodman1953ecological}. The ecological regression estimator for the difference in group means, $\widehat{D}_{ER}$, is defined as the estimated coefficient for $X_N$ in the weighted least squares regression of $Y_N$ on $X_N$, with weights based on the share of the population in each neighborhood:
\begin{equation*}
    \widehat{D}_{ER} = \frac{\sum\limits_{n\in\mathcal{N}}\,p_n\,\widetilde{Y}_n\,\widetilde{X}_n}{\sum\limits_{n\in\mathcal{N}}\,p_n\left(\widetilde{X}_n\right)^2}
\end{equation*}
where $\widetilde{X}_n=X_n-\frac{\sum_n p_n X_n}{\sum_n p_n}$ denotes the demeaned value of $X_n$, and similarly for $\widetilde{Y}_n=Y_n-\frac{\sum_n p_n Y_n}{\sum_n p_n}$. We will focus on the ecological regression estimator's asymptotic limit under an \emph{iid} sampling process:
\begin{equation*}
    D_{ER} = \frac{\text{Cov}(Y_N,X_N)}{\text{Var}(X_N)}
\end{equation*}
where $X_N = \mathbb{E}[X|N]$ and $Y_N = \mathbb{E}[Y|N]$. 

In turn, the ecological regression estimates for the group-level means are given by

\begin{equation*}
    \widehat{Y}_{ER}^0 = \sum\limits_{n\in\mathcal{N}}\,p_n\,Y_n - \left(\frac{\sum\limits_{n\in\mathcal{N}}\,p_n\,\widetilde{Y}_n\,\widetilde{X}_n}{\sum\limits_{n\in\mathcal{N}}\,p_n\,(\widetilde{X}_n)^2}\right)\,\left(\sum\limits_{n\in\mathcal{N}}\,p_n\,X_n\right)
\end{equation*}
and 
\begin{equation*}
   \widehat{Y}_{ER}^1 =     \sum\limits_{n\in\mathcal{N}}\,p_n\,Y_n + \left(\frac{\sum\limits_{n\in\mathcal{N}}\,p_n\,\widetilde{Y}_n\,\widetilde{X}_n}{\sum\limits_{n\in\mathcal{N}}\,p_n\,(\widetilde{X}_n)^2}\right)\,\left(\sum\limits_{n\in\mathcal{N}}\,p_n\,(1-X_n)\right)
\end{equation*}
which respectively converge to 

\begin{equation*}
     Y_{ER}^0 =  \mathbb{E}[Y_N] - \left(\frac{\text{Cov}(Y_N,X_N)}{\text{Var}(X_N)}\right)\,\mathbb{E}[X_N]
\end{equation*}
and
\begin{equation*}
    Y_{ER}^1 =  \mathbb{E}[Y_N] + \left(\frac{\text{Cov}(Y_N,X_N)}{\text{Var}(X_N)}\right)\,\left(1-\mathbb{E}[X_N]\right).\footnote{One can equivalently define the ecological regression estimates for the group means as the estimated coefficients of the weighted regression of $Y_N$ on $X_N$ and $1-X_N$.}
\end{equation*}

An alternative method for estimating group-level differences in ecological inference settings is the so-called Neighborhood Model (NM) \citep{freedman1991ecological}.\footnote{\citet{freedman1991ecological} proposes two model variants: the non-linear neighborhood model and the linear neighborhood. Our focus is on the former.} The neighborhood model estimators for the group-level means are weighted averages of the neighborhood-level outcomes, with weights given by the group's prevalence in the neighborhood:
\begin{equation*}
    \widehat{Y}^1_{NM} = \frac{\sum\limits_{n\in\mathcal{N}} p_nX_nY_n}{\sum\limits_{n\in\mathcal{N}} p_n X_n}
\end{equation*}

and 

\begin{equation*}
    \widehat{Y}^0_{NM} = \frac{\sum\limits_{n\in\mathcal{N}} p_n\,(1-X_n)\,Y_n}{\sum\limits_{n\in\mathcal{N}} p_n \,(1-X_n)\,}
\end{equation*}

The neighborhood model estimator for the difference in group means, $\widehat{D}_{NM}$, is correspondingly defined as
\begin{equation*}
    \widehat{D}_{NM} = \frac{\sum\limits_{n\in\mathcal{N}} p_nX_nY_n}{\sum\limits_{n\in\mathcal{N}} p_n X_n} - \frac{\sum\limits_{n\in\mathcal{N}} p_n (1-X_n) Y_n}{\sum\limits_{n\in\mathcal{N}} p_n (1-X_n)}
\end{equation*}

As above, we will focus on the neighborhood model estimators' asymptotic limits:
\begin{equation*}
 Y^1_{NM} = \frac{\mathbb{E}[X_N\,Y_N]}{\mathbb{E}[X_N]} 
\end{equation*}

\begin{equation*}
 Y^0_{NM} =  \frac{\mathbb{E}\left[ (1-X_N)\,Y_N \right]}{\mathbb{E} \left[ 1-X_N \right]}
\end{equation*}

and
\begin{equation*}
 D_{NM} = \frac{\mathbb{E}[X_N\,Y_N]}{\mathbb{E}[X_N]} -  \frac{\mathbb{E}\left[ (1-X_N)\,Y_N \right]}{\mathbb{E} \left[ 1-X_N \right]}
\end{equation*}

It will be useful in our subsequent results to observe that ecological regression and the neighborhood model estimates follow a close mechanical relationship.

\begin{lemma}[Relationship Between Neighborhood Model and Ecological Regression]\label{lemma:nm_er}

\begin{equation*}
    D_{NM} =\gamma \,D_{ER}
\end{equation*}
\end{lemma}
\noindent where $\gamma=\frac{\text{Var}(X_N)}{\text{Var}(X)}<1$, as above.

Whereas ecological regression and the neighborhood model yield point estimates for the group-level means and their difference, an alternative approach is to calculate the most extreme values of these quantities that are consistent with the observed data \citep{Duncan1953AnAT}. In EI settings, consistency with the data requires that, for each $n\in\mathcal{N}$, (i) $Y_n^0\in[\underline{Y},\overline{Y}]$; (ii) $Y_n^1\in[\underline{Y},\overline{Y}]$; and (iii) $X_n\,Y_n^1 + (1-X_n)Y_n^0=Y_n$.
The Method of Bounds (MOB) interval is defined by the minimum and maximum group-level means that satisfy these constraints. 

\begin{prop}[Method of Bounds for Group Means]\label{prop:mob_levels}
Suppose that for each $x$ and $n$, the conditional expectation of $Y$ with respect to $x$ and $n$ exists, $\underline{Y}\leq\mathbb{E}[Y|X=x,N=n]\leq \overline{Y}$ for some pair $(\underline{Y},\overline{Y})\in \mathbb{R}^2$. Define the following parameters:
    \begin{align*}
    Y_{MOB}^{1+} :&= \frac{\E\left[\min\left\{Y_n-\underline{Y}(1-X_n),\overline{Y}X_n\right\}\right\}]}{\E[X]}. \\
        Y_{MOB}^{0-} :&=\frac{\E[\max\left\{Y_n-\overline{Y}X_n,\underline{Y}(1-X_n)\right\}]}{1-\E[X]}\\
    Y_{MOB}^{1-} :&= \frac{\E[\max\left\{Y_n-\overline{Y}(1-X_n),\underline{Y}X_n \right\}]}{\E[X]}\\
    Y_{MOB}^{0+} :&= \frac{\E[\min\left\{Y_n-\underline{Y}X_n,\overline{Y}(1-X_n)\right\}]}{1-\E[X]}
\end{align*}

It follows that
\begin{align}
   Y^1 &\in \left[ Y^{1-}_{MOB} \,,\, Y^{1+}_{MOB} \right] \label{eqn:mob_means_y1}
\\  \text{ and } Y^0 &\in \left[ Y^{0-}_{MOB} \,,\, Y^{0+}_{MOB} \right] \label{eqn:mob_means_y0}
\end{align}

The bounds in (\ref{eqn:mob_means_y1}) and (\ref{eqn:mob_means_y0}) 
are sharp absent additional information.
\end{prop}
\vspace{6mm}
In Proposition \ref{prop:mob_levels} and throughout, we refer to an interval $[a\,,\,b]$ as a sharp bound for $Y^1$ if, for every $y^{1}\in \left[ a \,,\, b \right]$,
there exists some joint distribution of $(X,Y)$ that is consistent with the observed marginal distribution of ($p_n$, $X_n$, $Y_n$) and that implies $Y^1=y^{1}$, and similarly for bounds for $Y^0$ and $D$.

The same logic yields bounds for the difference in group means:

\begin{prop}[Method of Bounds for Difference in Group Means]\label{prop:mob_dif}
    Suppose that for each $x$ and $n$, the conditional expectation of $Y$ with respect to $x$ and $n$ exists, $\underline{Y}\leq|\mathbb{E}[Y|X=x,N=n]\leq \overline{Y}$ for some pair $(\underline{Y},\overline{Y})\in \mathbb{R}^2$. Define the following parameters:
        \begin{align*}
D_{MOB}^{+} &= \frac{\E[\min\left\{Y_n-\underline{Y}(1-X_n),\overline{Y}X_n\right\}]-\E[X]\E[Y]}{\Var[X]}
\\
\text{ and }   
D_{MOB}^- &= \frac{\mathbb{E}[Y](1-\E[X])  -\mathbb{E}[\min \{ Y_N-\underline{Y}X_n,\,\overline{Y}(1-X_N) \}] }{\Var(X)}
\end{align*}

It follows that
\begin{align}
   D\in \left[ D_{MOB}^- \,,\, D_{MOB}^+ \right] \label{eqn:mob_d}
\end{align}

The bounds in (\ref{eqn:mob_d}) are sharp absent additional information.
\end{prop}

In the next section, we consider the bias of these estimators and show that they define sharp bounds under various assumptions about the sign of the within- and between-group associations.

\section{Identification of the Difference in Group Means}

In this section we study how assumptions about the sign of the within- and between-group associations provide identifying power for learning about the difference in group means. We first establish that the bias of the ecological regression estimate for the difference in group means depends on the within-group association.\footnote{Propositions \ref{prop:bias_er} and \ref{prop:bias_nm} are the EI analogues to Proposition 1.1 and 1.2 in \citet{elzayn2024measuring}.} 

\begin{prop}[Bias of Ecological Regression for Difference in Group Means]\label{prop:bias_er}
\begin{equation*}
    D_{ER} - D=   \frac{ \delta_W}{\Var(X_N)}
\end{equation*}
\end{prop}

We have the following direct corollary:
\begin{corr}\label{corr:unbiased_er} The ecological regression estimator, $\widehat{D}_{ER}$, is unbiased if and only if $\delta_W=0$.
\end{corr}

\vspace{6mm}

When the within-group association is non-zero, ecological regression is biased because differences in the prevalence of groups across neighborhoods is conflated by other neighborhood-level contextual effects. This phenomenon was recognized as early as \citet{goodman1953ecological}; it is sometimes referred to as aggregation bias in the EI literature. 

Interestingly, the condition provided in Corollary 1 differs slightly from the condition typically described in the EI literature under which ecological regression is unbiased, which is that
\begin{equation}\label{eq:ecreg_bias_lit}\text{Cov}(Y^1_N,X_N)=\text{Cov}(Y^0_N,X_N)=0
\end{equation}
where $Y^x_N=\mathbb{E}[Y|N,X=x]$ for $x=0,1$.\footnote{
A stronger condition that is sometimes discussed, which implies \eqref{eq:ecreg_bias_lit}, is the constancy model, under which each group's outcomes are the same in each neighborhood, $Y^x_n=Y^x$ for all $n$ and for $x=0,1$.} Although condition \eqref{eq:ecreg_bias_lit} implies $\delta_W=0$, the converse is not true. This distinction is significant because condition \eqref{eq:ecreg_bias_lit} can be empirically tested with the residuals from the ecological regression, as proposed, for instance, by \citet{loewen1989recent} and \citet{gelman2001models}. However, $\delta_W$ may be zero (so that ecological regression is unbiased), even when the ecological regression residuals indicate that the conditional expectation function of $Y_N$ given $X_N$ is highly non-linear (indicating that Condition \eqref{eq:ecreg_bias_lit} is violated). Figure \ref{fig:example_goodman} provides an illustration. 

\begin{figure}[ht]\caption{Ecological Regression with Contextual Effects}
{    \centering   
    \includegraphics[width=0.6\textwidth]{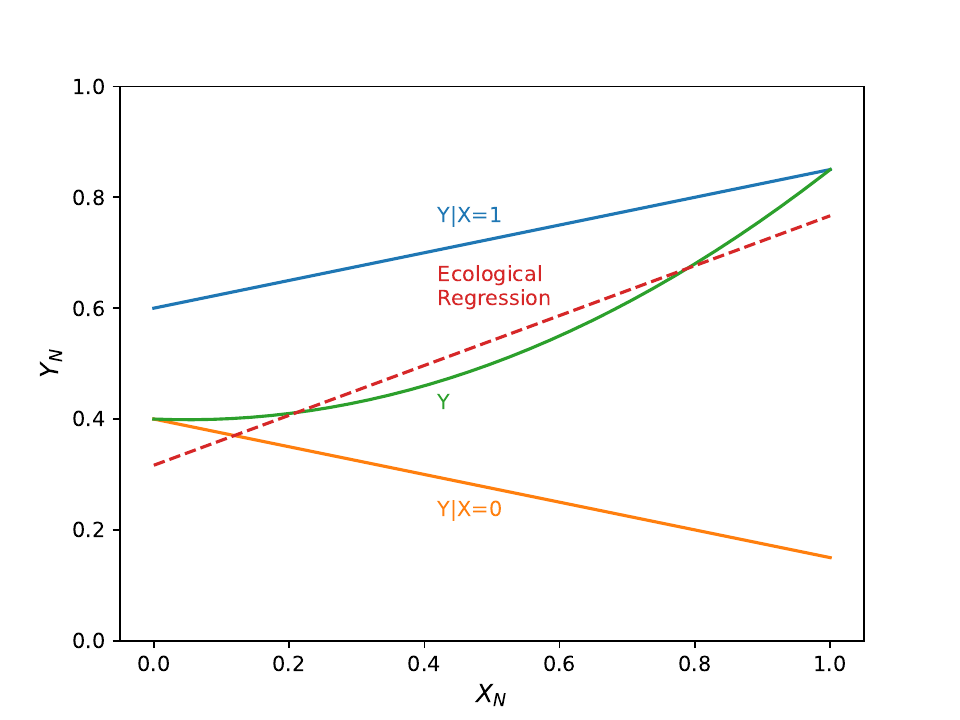}
    \vspace{-5pt} \\}
    {\footnotesize{Notes: The figure illustrates a setting in which $\text{Cov}(Y_N,X_N|X)\neq 0$ for each $X$, so that $\mathbb{E}[Y_N|X_N]$ is nonlinear. However, $\mathbb{E}[\text{Cov}(Y_N, X_N|X)] = \delta_W = 0$, so ecological regression is still unbiased. While the ecological regression line (red dashed) does not travel through $Y(0)$ or $Y(1)$ (green), the slope of the ecological regression line is equal to the difference between $Y(1)$ and $Y(0)$, so ecological regression is unbiased.}}
\label{fig:example_goodman}
\end{figure}
Our next result characterizes the identifying power of assumptions that sign the within-group (across neighborhood) association.

\begin{prop}[Bounds Based on the Sign of the Within-Group Association]\label{prop:cna_bounds}\hfill\break

(i) If $\delta_W\geq0$, then 
\begin{equation*} 
D\in\left[D_{MOB}^{-}\,,\, \text{Min}\left\{D_{MOB}^{+},\,D_{ER}\right\}\right]
\end{equation*}

(ii) If $\delta_W\leq0$, then 
\begin{equation*} 
D\in\left[\text{Max}\left\{D_{MOB}^{-},\,D_{ER}\right\}\,,\,D_{MOB}^{+} \right]
\end{equation*}

(iii) The bounds in (i) and (ii) are sharp in the absence of additional information.
\end{prop}

We now proceed analogously to derive bounds based on the sign of the between-group (within-neighborhood) association. First, Proposition \ref{prop:bias_nm} characterizes the bias of the neighborhood model:

\begin{prop}[Bias of the Neighborhood Model for the Difference in Group Means]\label{prop:bias_nm}
\begin{equation*}
    D_{NM} - D =  \frac{-\delta_B}{\Var(X)}
\end{equation*}
\end{prop}

\begin{corr}\label{corr:unbiased_nm}

The neighborhood model estimator, $\widehat{D}_{NM}$, is unbiased if and only $\delta_B=0$.
\end{corr}
\citet{freedman1991ecological} noted that the neighborhood model point-identifies $\mathbb{E}[Y\mid X,N]$ under a related but stronger condition, that there is no systematic difference in the outcome between groups within any neighborhood, i.e., $\mathbb{E}[Y|0,N]=\mathbb{E}[Y|1,N]$ for all $N$.

\vspace{6mm}

The following result uses Proposition \ref{prop:bias_nm} to derive sharp bounds for $D$ based on the sign of the between-group association.
 
\begin{prop}[Bounds Based on the Sign of the Between-Group Association]\label{prop:cga_bounds}\hfill \break

(i) If $\delta_B\geq0$, then 
\begin{equation*} 
D\in\left[D_{NM}\,,\,D_{MOB}^{+} \right]
\end{equation*}

(ii) If $\delta_B\leq0$, then 
\begin{equation*} 
D\in\left[D_{MOB}^{-}\,,\, D_{NM}\right]
\end{equation*}

(iii) The bounds in (i) and (ii) are sharp in the absence of additional information.
\end{prop}
\vspace{6mm}

It is interesting to compare the bounds in Proposition \ref{prop:cga_bounds} based on $\delta_B$ to the bounds in Proposition \ref{prop:cna_bounds} based on $\delta_W$. Both bounds make use of information about the sign of a conditional association to narrow the method of bounds interval in one direction. However, whereas the bounds based on $\delta_B$ are strictly within the interior of the method of bounds interval, the same is not true of $\delta_W$. The explanation for the difference is that the estimates from ecological regression are not guaranteed to be feasible, whereas the estimates from the neighborhood model are necessarily within the method of bounds interval. 

The bounds in Proposition \ref{prop:cga_bounds} are conceptually related to the ``bounded variation'' approach considered in \citet{manski2018credible}, which can entail imposing an assumption on the sign of $\mathbb{E}[Y|X=1,N]-\mathbb{E}[Y|X=0,N]$ to derive bounds for $\mathbb{E}[Y|X,N]$.\footnote{More generally, \citet{manski2018credible} considers identification of $p(Y,X,N)$ given assumptions that constrain the possible values of $p(Y|X=1,N)-p(Y|X=0,N)$. When $Y$ is binary, this term is equal to $\text{Cov}(Y,X|N)$, the analogue of $\delta_B$ for a specific $N$.} Because our focus is on $\mathbb{E}[Y|X]$ rather than $\mathbb{E}[Y|X,N]$, Proposition \ref{prop:cga_bounds} requires only that this sign restriction hold \emph{on average} across values of $N$. See Section \ref{sec:long_reg} for additional discussion.

Whereas the bounds in Propositions \ref{prop:cna_bounds} and \ref{prop:cga_bounds} are based on assumptions about the respective signs of $\delta_W$ and $\delta_B$, in some settings it may be more credible to assume that $\delta_W$ and $\delta_B$ share the same sign as one another, without taking a stance on what that shared sign is. For example, the researcher may have theoretical or institutional reasons to expect that contextual effects reinforce the individual-level (within-neighborhood) difference, such as through mutually reinforcing social norms. We refer to this condition as contextual reinforcement, and formally define it as follows:

\begin{definition}
   Contextual Reinforcement \emph{is satisfied if} $\delta_W\cdot\delta_B\geq0$. 
\end{definition}

\begin{prop}(Identification of Group Differences with Contextual Reinforcement)\label{prop:bounds_cr}

If $\delta_W\cdot\delta_B\geq0$, then:

(i) If $D_{ER}\geq D_{NM}$, then 
\begin{equation*}
0 \leq  D_{NM} \leq D \leq \text{Min}\left\{D_{MOB}^{+}\,,\, D_{ER}\right\}
\end{equation*}

(ii) If $D_{ER}\leq D_{NM}$, then 
\begin{equation*}
\text{Max}\left\{D_{MOB}^{-}\,,\, D_{ER}\right\} \leq D \leq D_{NM} \leq 0
\end{equation*}

(iii) The bounds in (i) and (ii) are sharp in the absence of additional information.
\end{prop}
\vspace{6mm}

\begin{corr}\label{lem:sign_D_contextual_reinf}If contextual reinforcement holds, then 
\begin{equation*}
    D\geq 0 \iff D_{ER}\geq D_{NM}
\end{equation*}
\end{corr}

The following Theorem summarizes our  results so far and provides corresponding results for cases in which both conditional associations have known signs but contextual reinforcement does not hold.

\begin{thm}[Identification of $D$ with Monotone EI]
\label{thm:identification_D} Table \ref{tab:summary_d} provides sharp bounds for $D$ under the specified information. If it is known that contextual reinforcement holds, the bounds on $D$ in one of the shaded table cells will apply. 
\end{thm}

\section{Identification of Group Means}

This section provides results relating to the use of monotone ecological inference for identifying the group means, $Y^0$ and $Y^1$. 

The ecological regression estimates for the group means are biased under the same condition, and in a related manner, as the ecological regression estimate for the difference in group means described in Proposition \ref{prop:bias_er}. A similar parallel exists for the neighborhood model estimates of the group means compared to the neighborhood level estimate for the difference in group means. The nature of these relationships are formalized in the following two propositions.

\begin{prop}[Bias of Ecological Regression for Group Means]\label{prop:bias_er_levels}
\begin{equation*}
    Y^1_{ER} - Y^1=   \frac{ \delta_W \, (1-\mathbb{E}[X])}{\Var(X_N)}
\end{equation*}
\begin{equation*}
    Y^0_{ER} - Y^0=   -\frac{ \delta_W \, \mathbb{E}[X]}{\Var(X_N)}
\end{equation*}
\end{prop}

\begin{prop}[Bias of Neighborhood Model for Group Means]\label{prop:bias_nm_levels}
\begin{equation*}
    Y^1_{NM} - Y^1=   -\frac{ \delta_B}{\mathbb{E}[X]}
\end{equation*}
\begin{equation*}
    Y^0_{NM} - Y^0=   \frac{ \delta_B}{1-\mathbb{E}[X]}
\end{equation*}
\end{prop}

Applying these results in the manner employed in the prior subsection yields the following identification results on $Y^1$ and $Y^0$.

\begin{thm}[Identification of $Y^1$ with Monotone EI]\label{thm:identification_Y1}

 Table \ref{tab:summary_y1} provides sharp bounds for $Y^1$ under the specified information. If it is known that contextual reinforcement holds, the bounds on $Y^1$ in one of the shaded table cells will apply. 
\end{thm}

\begin{thm}[Identification of $Y^0$ with Monotone EI]
\label{thm:identification_Y0}

 Table \ref{tab:summary_y0} provides sharp bounds for $Y^0$ under the specified information. If it is known that contextual reinforcement holds, the bounds on $Y^0$ in one of the shaded table cells will apply. 
\end{thm}

Finally, our discussion has so far assumed that group membership is binary. Suppose instead that there are $G+1$ discrete groups, indexed by $g\in\{0,1,...,G\}.$ To identify $\mathbb{E}[Y|G=g]$, we can define $X^g\in\{0,1\}$ to indicate whether or not an individual belongs to group $g$, $$X^g=1 \iff G=g$$ and then apply the bounds in Table 2 with respect to $X^g$ instead of $X$. Note that unlike the binary case, the relevant conditional associations will differ for each group. 
For example, the within-group association that determines the bias of the ecological regression estimator for $\mathbb{E}[Y|G=g]$ is given by $$\delta_W^g = \mathbb{E}[\text{Cov}(Y,X_N^g|X^g)],$$
where $X_N^g=\mathbb{E}[X^g \mid N]$. Similarly, the between-group association that determines the bias of the neighborhood model estimator for $\mathbb{E}[Y|G=g]$ is given by $$\delta_B^g = \mathbb{E}[\text{Cov}(Y,X^g|N)].$$ 

\pagebreak
\begin{landscape}
\renewcommand*\arraystretch{1.5}
\begin{table}
  \caption{Identification of the Difference in Group Means%
  }
  \label{tab:summary_d}
    \centering
    \begin{adjustbox}{width=1.2\textwidth,center=\textwidth}
    \begin{tabular}{c|c|c|c|c}         \backslashbox{Within}{Between}
&$\delta_B=\,? %
$    & $\delta_B\geq0$ &$\delta_B \leq0$
 &$\delta_B=0$\\

         \hline
     
        $\delta_W =\,?%
        $  &$D \in \left[D_{MOB}^{-},D_{MOB}^+\right]$  
        & $D\in \left[D_{NM},D_{MOB}^+\right]$ & $D\in\left[D_{MOB}^-,D_{NM}\right]$ & $D=D_{NM}$\\
     $\delta_W \geq 0$  & $D\in\left[D_{MOB}^{-}\,,\, \Min\left\{D_{MOB}^{+},\,D_{ER}\right\}\right]$ &   \cellcolor{blue!25}$0 \leq  D_{NM} \leq D \leq \Min \left\{D_{MOB}^{+}\,,\, D_{ER}\right\}$ &$D \in \left[ D_{MOB}^{-} , \min\left\{D_{NM},D_{ER},D_{MOB}^{+}\right\} \right]$&$D=D_{NM} $  \\
    $\delta_W \leq 0 $  &$D\in\left[\Max\left\{D_{MOB}^{-},\,D_{ER}\right\}\,,\,D_{MOB}^{+} \right] $ & $D 
     \in \left[ \Max \left\{D_{MOB}^{-},D_{ER},D_{NM}\right\}, D_{MOB}^+\right]$  &\cellcolor{blue!25}$\Max\left\{D_{MOB}^{-}\,,\, D_{ER}\right\} \leq D \leq D_{NM} \leq 0 $&$D=D_{NM}$\\
    $\delta_W=0$&    $D=D_{ER}$ & $D=D_{ER}$ & $D=D_{ER}$&$D=0$ \\
    \end{tabular}
        \end{adjustbox}
    \label{tab:theorem}
\end{table}

\renewcommand*\arraystretch{1.5}
\begin{table}
  \caption{Identification of the Group 1 Mean%
  }
  \label{tab:summary_y1}
    \centering
    \begin{adjustbox}{width=1.2\textwidth,center=\textwidth}
    \begin{tabular}{c|c|c|c|c}         \backslashbox{Within}{Between}
&$\delta_B =\,?%
$    & $\delta_B\geq0$ &$\delta_B \leq0$
 &$\delta_B=0$\\

         \hline
     
        $\delta_W =\,? %
        $ 
        & $Y^1 \in \left[Y_{MOB}^{1-} , Y_{MOB}^{1+}\right]$ 
        & $Y^1 \in \left[Y_{NM}^{1}, Y_{MOB}^{1+}\right]$
        & $Y^1 \in \left[Y_{MOB}^{1-}, Y_{MOB}^{1+}\right]$
        &$Y^1=Y^{NM}$\\
      $\delta_W\geq0$
      & $Y^1 \in \left[Y_{MOB}^{1-}, \Min\left\{Y_{ER}^{1},Y_{MOB}^{1+}\right\}\right]$
      &\cellcolor{blue!25}$Y_{NM}^{1}\leq Y^{1} \leq \Min\left\{Y_{ER}^{1},Y_{MOB}^{1+}\right\}$
      &$Y^1 \in \left[Y_{MOB}^{1-}, \Min\left\{Y_{NM}^{1},Y_{ER}^1,Y_{MOB}^{1+}\right\}\right]$
      &$Y^1=Y^{NM}$  \\
    $\delta_W \leq 0 $  
    &$Y^1 \in \left[\Max\left\{Y_{MOB}^{1-},Y_{ER}^{1}\right\}, Y_{MOB}^{1+} \right]$
    &$Y^1 \in \left[\Max\left\{Y_{ER}^{1},Y_{NM}^1,Y_{MOB}^{1-}\right\}, Y_{MOB}^{1+}\right]$
    &\cellcolor{blue!25}$\Max\left\{Y_{ER}^{1},Y_{MOB}^{1-}\right\} \leq Y^1\leq Y_{NM}^1$
    &$Y^1=Y^{NM}$\\
    $\delta_W=0$
    &  $Y^{1}=Y^{ER}$  
    &$Y^{1}=Y^{ER}$
    &$Y^{1}=Y^{ER}$
    &$Y^{1}=Y^{0}$ \\
    \end{tabular}
        \end{adjustbox}
    \label{tab:theorem}
\end{table}

\renewcommand*\arraystretch{1.5}
\begin{table}
  \caption{Identification of the Group 0 Mean%
  }
  \label{tab:summary_y0}
    \centering
    \begin{adjustbox}{width=1.2\textwidth,center=\textwidth}
    \begin{tabular}{c|c|c|c|c}         \backslashbox{Within}{Between}
&$\delta_B = \,?%
$    & $\delta_B\geq0$ &$\delta_B \leq0$
 &$\delta_B=0$\\

         \hline
     
        $\delta_W = \,?%
        $ &$Y^0 \in \left[Y_{MOB}^{0-}, Y_{MOB}^{0+}\right]$ 
        &$Y^0 \in \left[Y_{MOB}^{0-}, Y_{NM}^{0}]\right]$
        &$Y^0 \in \left[Y_{NM}^0, Y_{MOB}^{0+}\right]$
        &$Y^0=Y_{NM}^{0}$\\
      $\delta_W\geq0$
      &$Y^0 \in \left[\Max\left\{Y_{ER}^0,Y_{MOB}^{0-}\right\}, Y_{MOB}^{0+}\right]$ 
      & \cellcolor{blue!25} $\Max\left\{Y_{ER}^{0}, Y_{MOB}^{0-} \right\} \leq Y^{0} \leq Y_{NM}^{0}$ 
      &$Y^0 \in \left[\max\left\{Y_{MOB}^{0-},Y_{ER}^,Y_{NM}^0\right\} Y_{MOB}^{0+}\right]$
      &$Y^0=Y_{NM}^{0}$  \\
    $\delta_W \leq 0 $  
    & $Y^0 \in \left[Y_{MOB}^{0-}, \Min\left\{Y_{ER}^0,Y_{MOB}^{0+}\right\}\right]$
    & $Y^0 \in \left[Y_{MOB}^{0-}, \Min\left\{Y_{NM}^{0},Y_{ER}^0,Y_{MOB}^{0+}\right\}\right]$
    &\cellcolor{blue!25} $Y_{NM}^0\leq Y^{0}\leq \Min\left\{Y_{ER}^{0},Y_{MOB}^{0+}\right\}$
    &$Y^0=Y_{NM}^{0}$\\
    $\delta_W=0$
    &  $Y^{0}=Y_{ER}^0$  
    &$Y^{0}=Y_{ER}^0$
    &$Y^{0}=Y_{ER}^0$
    &$Y^{0}=Y^{1}$ \\
    \end{tabular}
        \end{adjustbox}
    \label{tab:theorem}
\end{table}

\end{landscape}

\pagebreak

\section{Identification of Neighborhood-Specific Group Means}\label{sec:long_reg}

Our focus so far has been on identification of the overall group means, $Y^x=\mathbb{E}[Y|X=x]$, as well as the difference in overall group means, $D=Y^1-Y^0$. In some settings, a researcher may seek to learn a group mean within a specific neighborhood, $Y_n^x=\mathbb{E}[Y\mid N=n,X=x]$, or the difference in group means within that neighborhood, $D_n=Y^1_n-Y^0_n$.\footnote{Using the terminology of \citet{cross2002regressions}, the overall group mean corresponds to the ``short regression'' of $Y$ on $X$ whereas the neighborhood-specific group mean corresponds to the ``long regression'' of $Y$ on $X$ and $N$.} In this section, we consider the identifying power of monotone ecological inference for studying neighborhood-specific group means.

As a starting point, note that the method of bounds provides sharp bounds on $Y_n^x$ and $D_n$ for each $n$ and $x$ under the same assumptions imposed by Propositions \ref{prop:mob_levels} and \ref{prop:mob_dif}. We will refer to the endpoints of these intervals for a specific neighborhood $n$ using the following notation:
\begin{equation*}
    Y_n^x \in \left[ Y_{MOB,n}^{x-} \,,\, Y_{MOB,n}^{x+}  \right]
\end{equation*}
and
\begin{equation*}
    D_n \in \left[ D_{MOB,n}^{-} \,,\, D_{MOB,n}^{+}  \right]
\end{equation*}

where $Y_{MOB,n}^{1+} = \min\left\{\frac{Y_n}{X_n},\overline{Y}\right\}$, $Y_{MOB,n}^{0-} = \frac{Y_n-\min\{Y_n,X_n\overline{Y}\}}{1-X_n}$, $Y_{MOB,n}^{0+}= \min\left\{\frac{Y_n}{1-X_n},\overline{Y}\right\}$, $Y_{MOB,n}^{1-} = \frac{Y_n - \min\{Y_n, \overline{Y}(1-X_n)\}}{X_n}$, $D_{MOB,n}^{+} = Y_{MOB,n}^{1+} - Y_{MOB,n}^{0-}$, and $D_{MOB,n}^{-} = Y_{MOB,n}^{1-}-Y_{MOB,n}^{0+}$. 

Monotone ecological inference entails sharpening these bounds at the neighborhood-level by imposing sign restrictions on various aspects of the unobserved individual-level relationship between $Y$, $X$, and $N$. Consider first the neighborhood-level analog to the between-group association,
\begin{align*}
    \delta_{B,n}:= \Cov(Y,X|N=n)
\end{align*}
the identifying power of which was also considered by \citet{manski2018credible}. Whereas $\delta_B$ describes the \emph{average} between-group variation within neighborhoods, $\delta_{B,n}$ depends on the between-group variation for a specific neighborhood.

Signing $\delta_{B,n}$ for a specific neighborhood $n$ yields sharp bounds on $D_n$ and $Y_n^X$.

\begin{prop}[Bounds Based on the Sign of the Neighborhood-Specific Between-Group Association]\label{prop:cga_long_bounds} \hfill \break
(i) If $\delta_{B,n}\geq0$ for some $n\in\mathcal{N}$ then:
\begin{align*}
    D_{n} \in [0, D_{MOB,n}^+],
\\
Y_n^1 \in [Y_n, Y_{MOB,n}^{1+}]\text{,}
\\   \text{ and }  Y_{n}^0 \in [Y_{MOB,n}^{0-},Y_n].
\end{align*}

(ii) If $\delta_{B,n}\leq0$ for all $n\in\mathcal{N}$ then:
\begin{align*}
    D_n \in [D_{MOB,n}^{-},0],
\\
Y_n^1 \in [Y_{MOB,n}^{1-}, Y_n],
\\ \text{ and }    Y_n^0 \in [Y_n, Y_{MOB,n}^{0+}].
\end{align*}

(iii) The bounds in (i) and (ii) are sharp in the absence of additional information.
\end{prop}

We next consider the identifying power of assumptions on the neighborhood-specific analog to the within-neighborhood association. Let the function $\mu_x(x_n)$ denote the conditional expectation of $Y$ for a member of group $x$ in a neighborhood with group-prevalence $x_n$
\begin{equation*}
    \mu_{x}(x_n)= \mathbb{E}[Y \mid X=x,\,X_n=x_n]
\end{equation*}
We restrict our focus to neighborhoods with group-prevalence values for which the first derivative of $\mu_x(\cdot)$ exists for $x=0$ and $x=1$, and we denote those derivatives by $\mu'_x(\cdot)$
\begin{equation*}
    \mu'_{x}(x_n)= \lim_{u\rightarrow0} \, \frac{1}{u} \, \big(\mathbb{E}[Y \mid X=x,\,X_n=x_n + u] - \mathbb{E}[Y \mid X=x,\,X_n=x_n] \big)
\end{equation*}

For neighborhood $n$, define the local within-neighborhood association $\delta_{W,n}$ as
\begin{equation*}
    \delta_{W,n} = X_n \,\mu'_1(X_n) + (1-X_n)\,\mu'_0(X_n)
\end{equation*}
where $X_n=\mathbb{E}[X|N=n]$. Conceptually, $\delta_{W,n}$ captures the within-group association between the mean of the outcome, $Y$, and the group-prevalence of the neighborhood, $X_N$, among neighborhoods with similar levels of group-prevalence. It differs from $\delta_W$ in that $\delta_W$ depends on a summary measure of the relationship between $Y$ and $X_N$ across all neighborhoods, whereas $\delta_{W,n}$ reflects the ``local'' relationship between $Y$ and $X_N$ among a set of neighborhoods with similar levels of group prevalence.

Although $\mu_0$ and $\mu_1$ are unobserved, we do observe the overall regression function, $\mu(x_n):=\mathbb{E}[Y \mid X_n=x_n]$, which is the mixture of the two group-specific regression functions
\begin{equation}\label{eq:yb_accounting}
    \mu(x_n) = x_n \,\mu_1(x_n) + (1-x_n) \, \mu_0(x_n).
\end{equation}
We also observe the derivative of the overall regression function, $\mu'(x_n)$, which, by construction, is guaranteed to exist.

Differentiating \eqref{eq:yb_accounting} yields
\begin{align*}
    \mu'(x_n) &=\mu_1(x_n) - \mu_0(x_n) + x_n \,\mu'_1(x_n) + (1-x_n) \, \mu'_0(x_n)\\
    &= D_n + \delta_{W,n}
\end{align*}
Thus, given knowledge of $\delta_{W,n}$, the derivative of the conditional expectation function, $\mu'(x_n)$, provides information about the group means in neighborhoods with group-prevalence $x_n$. 

In addition to assuming knowledge of the sign of $\delta_{W,n}$, it will be convenient to assume that each neighborhood with the same prevalence $X_n$ has the same average outcome - i.e. $\mathbb{E}[Y|X=x,N=n] = \mathbb{E}[Y|X=x, X_{N}=X_n]$. 
\begin{assmp}\label{assmp:exclusion}
    For every neighborhood $n$, \begin{align*}
        \mathbb{E}[Y|X=x,N=n] = \mathbb{E}[Y|X=x, X_{N}=X_n]
    \end{align*}
\end{assmp}
In words, Assumption \ref{assmp:exclusion} requires that any two neighborhoods that have the same group-prevalence will also share the same mean outcomes by group (at least in expectation). Under this assumption, knowledge of the sign of $\delta_{W,n}$ facilitates identification as follows:
\begin{prop}\label{prop:cna_long_bounds} \hfill \break
(i) Suppose that $\delta_{W,n} \geq0$. Then:\\
\begin{align*}
        D_n &\in\left[ D_{MOB,n}^{-} \min\left\{\mu'(X_n), D_{MOB,n}^+\right\}\right],
\\     Y_n^1 &\in [Y_{MOB,n}^{1-}, \min \left\{Y_{MOB}^{1+}, Y_n + (1-X_n)\cdot \mu'(X_n)\right\}],   
\\ \text{ and }   
Y_n^0 &\in [\max\left\{Y_{MOB,n}^{0-}, Y_n-X_n\cdot \mu'(X_n)\right\}, Y_{MOB,n}^{0+}].
    \end{align*}\\
(ii) Suppose that $\delta_{W,n}\leq 0$. Then:\\
\begin{align*}
    D_{n} &\in \left[ \max\left\{\mu'(X_n),D_{MOB,n}^{-}\right\} , D_{MOB,n}^+
    \right],
\\  Y_n^1 &\in [\max\left\{Y_{MOB,n}^{1-},Y_n+(1-X_n) \cdot \mu'(X_n)\right\}, Y_{MOB,n}^{1+}],  
\\ \text{ and }    Y_n^0 &\in [Y_{MOB,n}^{0-},\min\left\{Y_{MOB,n}^{0+},Y_n-X_n\cdot \mu'(X_n)\right\}].
\end{align*}\\
(iii) The bounds in (i) and (ii) are sharp in the absence of further information.
\end{prop}

When Assumption 1 does not hold, the bounds in Proposition \ref{prop:cna_long_bounds} identify the group means or differences for the set of neighborhoods with the same group prevalence as the target neighborhood $n$. We formalize this claim and provide a proof in Appendix Proposition \ref{prop:cna_long_bounds_gen}.

Finally, consider the possibility that the researcher has information not about the individual signs of $\delta_{W,n}$ and $\delta_{B,n}$ but rather that these two (local) associations share the same sign as one another. This local analog to contextual reinforcement can substantially aid in identification:
\begin{prop}\label{prop:cr_local} Suppose that $\delta_{B,n} \cdot\delta_{W,n}\geq0$. Then either:\hfill\break

(i) $\mu'(X)\geq0$, \text{and } 
        \begin{align*} D_n &\in \left[
            0, \min\left\{\mu'(X_n),D_{MOB,n}^{+}\right\}\right],
\\
Y_n^1 &\in \left[Y_n, \min\left\{Y_n + (1-X_n)\mu'(X_n), Y_{MOB}^{1+}\right\}\right],
\\
\text{ and } Y_n^0 &\in \left[\max\left\{Y_{MOB,0}^{-},Y_n- X_n \mu'(X_n) \right\}, Y_n \right];
        \end{align*}
        or\\
        
(ii) $\mu'(X_n)\leq 0$, and 
        \begin{align*}
            D_n &\in [\max\left\{\mu'(X_n),D_{MOB,n}^-\right\},0],
\\
Y_n^1 &\in [\max\left\{Y_{MOB,n}^{1-},Y_n+(1-X_n)\mu'(X_n)\right\},Y_n],
\\
\text{ and }
Y_n^0 &\in [Y_n,\min\left\{Y_n-X_n\mu'(X_n), Y_{MOB,n}^{0+}\right\}].
        \end{align*}
        
(iii) The bounds in (i) and (ii) are sharp.
\end{prop}

\section{Empirical Application}

To illustrate the benefits of monotone ecological inference in a concrete setting, we investigate the question of partisan polarization in COVID-19 vaccine uptake. This question has been of acute interest to  policymakers, academics, and the media, both for COVID-19 response specifically and broader questions of trust in public health authorities \citep{milligan, collins}.\footnote{\citet{jones2022partisanship} provide a helpful conceptual discussion and review of the literature.}  Unfortunately, most research in the U.S. context has had to rely on ecological inference, as joint individual data on vaccination status and partisan membership is not available for a large nationally representative set of individuals \citep[e.g.,][]{albrecht2022vaccination, ye2023exploring}.\footnote{An alternative methodological approach is to collect individual-level survey data on vaccine uptake and political affiliation; such research designs avoid the need for ecological inference but often face limitations based on sample size, representativeness, and response bias.}

Our primary data set consists of county-level data from 3,115 counties on COVID-19 vaccination uptake and partisanship. We measure vaccination uptake as the share of county residents who had received one or more COVID-19 vaccination as of December 31, 2021, obtained from the \citet{CDC_DATA}. We measure partisanship as the fraction of voters in the county who cast their ballot for the Republican candidate in the 2020 presidential election, obtained from the \citet{DVN/VOQCHQ_2020}.\footnote{We focus on this measure of partisan ideology rather than Republican party voter registration to be consistent with the prior EI research on the topic. A second reason is that a meaningful share of voters are registered as independent, and the ideologies associated with voters in this group is likely to vary widely across different counties around the country.} 

Using our previous notation, $Y$ indicates whether an individual is vaccinated, $X$ indicates whether an individual is Republican, and $N$ indicates the county in which the individual resides. Our goal is to use the county-level data to estimate the partisan vaccination gap, which we define as the difference in the mean vaccination rate of Republicans relative to Democrats and third party voters, $D=\mathbb{E}[Y|X=1] - \mathbb{E}[Y|X=0]$. 

Figure \ref{fig:republican_vaccination_county} plots the binned county-level data. The figure shows a clear downward trend: counties with more Republicans tend to have lower vaccination rates. The pattern is consistent with the possibility that Republicans are vaccinated at lower rates than Democrats. However, this interpretation is potentially subject to the ecological fallacy; counties with more Republicans may have lower vaccination rates for reasons unrelated to partisan composition \citep{ye2023exploring}. Indeed, the method of bounds interval for the partisan vaccination gap ranges from -77.4 to 52.8 percentage points. The width of the MOB interval implies that the county-level data does not provide much information about the magnitude or even direction of differences in vaccination rates between Republicans and Democrats, at least without further assumptions. Point identification approaches also diverge sharply: the neighborhood model and ecological regression imply respective partisan vaccination gaps of -5.5 and -47.9 percentage points.

\begin{figure}[ht]
    \caption{County Vaccination Rate by County Republican Vote Share}
{\centering      
\includegraphics[width=0.9\textwidth]{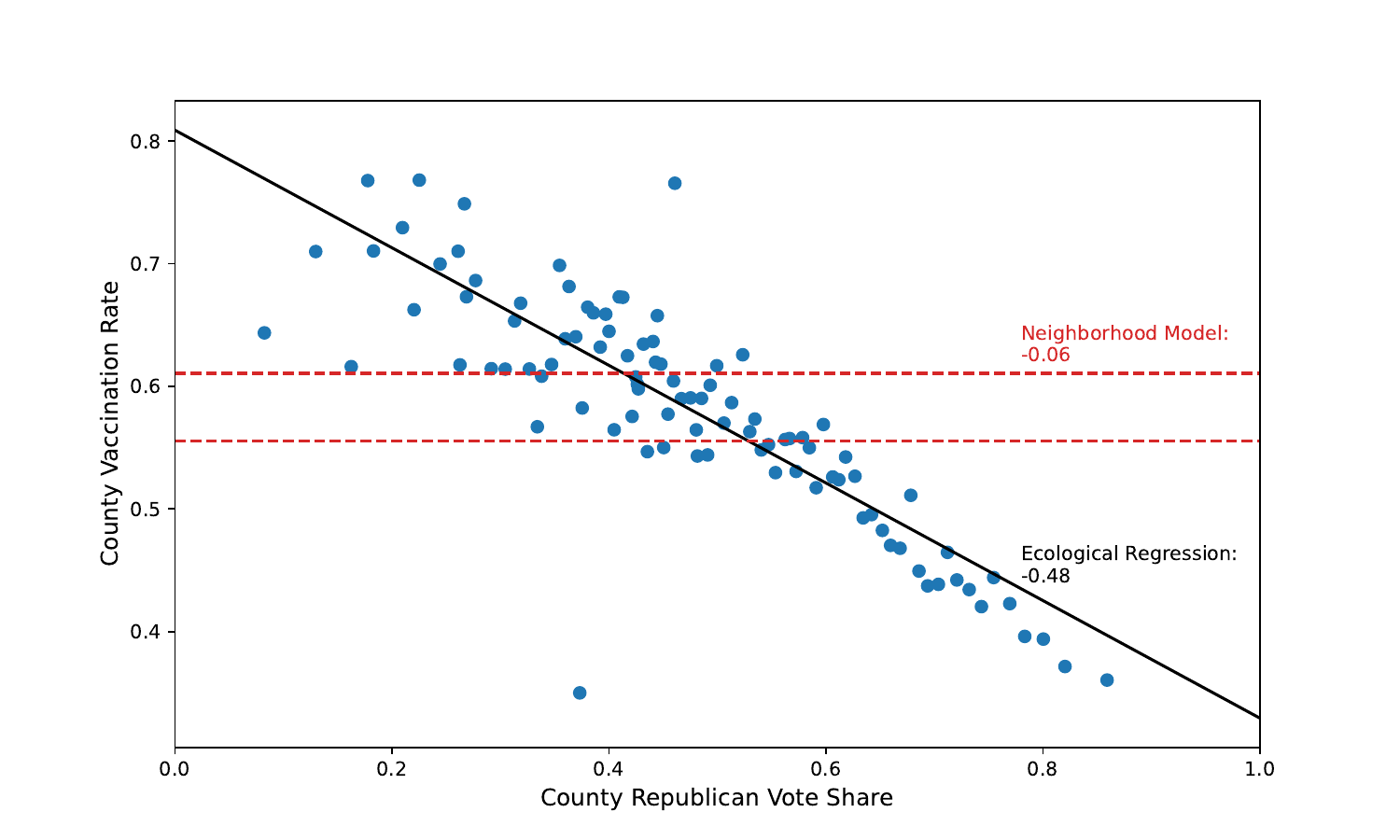}
    \vspace{-5pt} \\
    }
{\footnotesize{Notes: The figure reports county-level COVID-19 vaccination rates by the share of voters in the county who voted for the Republican candidate in the 2020 presidential election. Counties are grouped into 100 equal-population bins. The neighborhood model estimates for the mean vaccination rates among Republicans and non-Republicans are respectively denoted by the lower and upper red dotted lines. The ecological regression line is in black. 
    }}
\label{fig:republican_vaccination_county}
\end{figure}

Sharpening identification through monotone ecological inference involves making assumptions about two conditional associations. First, the between-group association, $\mathbb{E}\left[\text{Cov}(Y,X|N)\right]$, refers to differences in vaccination uptake between Republicans and Democrats living in the same county. Prior research provides some basis for expecting the between-group association to be negative; there are well-documented partisan differences in information sources that are not fully mediated through neighborhood \citep{iyengar2009red, peterson2021partisan}, and the prominent Republican politicians featured on more conservative media outlets were more likely to espouse anti-vaccination beliefs and/or downplay the health risks associated with the COVID-19 virus \citep{hornsey2020donald, gollwitzer2020partisan, albrecht2022vaccination}. Second, the within-group association, $\mathbb{E}\left[\text{Cov}(Y,X_N|X)\right]$, refers to differences in vaccination uptake among individuals of the same political party who live in neighborhoods with differing concentrations of Republican voters.\footnote{That is, the within-group association refers to vaccination uptake gaps between Republicans living in counties with more Republicans compared to Republicans living in counties with fewer Republicans, and between Democrats living in counties with more Republicans compared to Democrats living in counties with fewer Republicans.} Like the between-group association, there is some reason to believe that the within-group association is negative. For example, Republican counties tend to be lower income and more rural (Figure \ref{fig:rural_income_republican}), which are factors associated with lower access to public health services like vaccinations \citep{sun2022rural, hernandez2022disparities, parolin2022role}.
 
More generally, there are reasons to expect that the two conditional associations share the same sign as one another, whether that sign is positive or negative. One mechanism through which such contextual reinforcement may operate is network effects, such as social norms or peer effects. In particular, people's health behaviors are known to be influenced by the people around them \citep{sato2019peer, klaesson2023social}, and in Republican counties, a larger share of the people with whom one interacts are likely to be Republican. Thus if the Republicans in a neighborhood tend to be more skeptical of COVID-19 vaccinations (i.e., the between-group association is negative), that is likely to reduce the vaccination rate among both Democrats and Republicans living in that neighborhood. In the words of one author, ``In many communities, wearing a mask or getting a [COVID-19] vaccine became a political statement, with many Republicans arguing that these actions violated their individual freedoms and were unnecessary anyway'' \citep{albrecht2022vaccination}. Along similar lines, for many people, vaccine uptake may depend in part on local policies, such as whether vaccines are mandated for public sector employees \citep{howard2022association}. Thus, if Republicans exhibit more vaccine hesitancy, we would expect that counties in which more Republicans live would be more likely to elect leaders that do not adopt pro-vaccine policies, leading to lower  vaccine rates for county residents, whether Democrat or Republican.\footnote{Analogously, \citet{patterson2022politics} finds that states with Republican governors were less likely to quickly adopt stay-at-home orders during the pandemic, and that individuals in those states exhibited less social distancing as a result.}

The foregoing discussion provides a theoretical basis for the contextual reinforcement assumption in this setting. We empirically validate the assumption by drawing on an auxiliary dataset that contains individual-level data on vaccination status, political party registration, and neighborhood. We construct this dataset by matching a national dataset of voter registration records \citep{l2data2024}, which contain individual-level data on political party, to a large dataset of electronic health care records \citep{balraj2023american}, which contain individual-level data on COVID-19 vaccination status. Appendix \ref{sec:app_application} provides a further description of the underlying data sources and of our matching procedure. The final matched dataset contains approximately 1.3 million registered Republicans and Democrats in 2,576 counties and 49 states, plus the District of Columbia.\footnote{While this is a large linked dataset, most of the literature has not been able to secure such individual-level data due to data restrictions, hence relying primarily on aggregate data \citep[see][]{albrecht2022vaccination, ye2023exploring}.}

Figure \ref{fig:context_reinf} uses the auxiliary data to plot vaccination rates by (binned) county-level Republican vote share, separately for Republicans and Democrats. The figure provides visual support for contextual reinforcement: the Republican bins tend to lie below the Democratic bins with similar partisan makeup (so that the between-group association is negative) and both the Republican bins and Democratic bins exhibit a  downward sloping trend (so that the within-group association is also negative). Formal statistical tests regarding the sign of these quantities yield the same conclusion (see Table \ref{tab:zero_tie_unitwt}). Based on these results, we adopt the contextual reinforcement assumption to interpret the county-level analyses.\footnote{A limitation of this analysis for assessing contextual reinforcement is that the set of individual included in the auxiliary dataset may not be representative of the overall population due to the nature of selection into our electronic health records data or non-random match rates with the voter registration data. We obtain similar results when we replicate the analysis using alternative matching criteria to construct the auxiliary sample (Table \ref{tab:random_tie_unitwt}) and when we re-weight the auxiliary data based on individuals' observable characteristics to more closely match the national population (Table \ref{tab:zero_tie_rewt}). A different limitation is the potential for discordance between the partisanship measure contained in the auxiliary data (party registration) and the one employed in our main county-level analysis (presidential vote share).}

\begin{figure}[ht]
\caption{Vaccination Rate by Vote Share and Political Party Membership}
    \begin{centering}
    \includegraphics[width=0.8\textwidth]{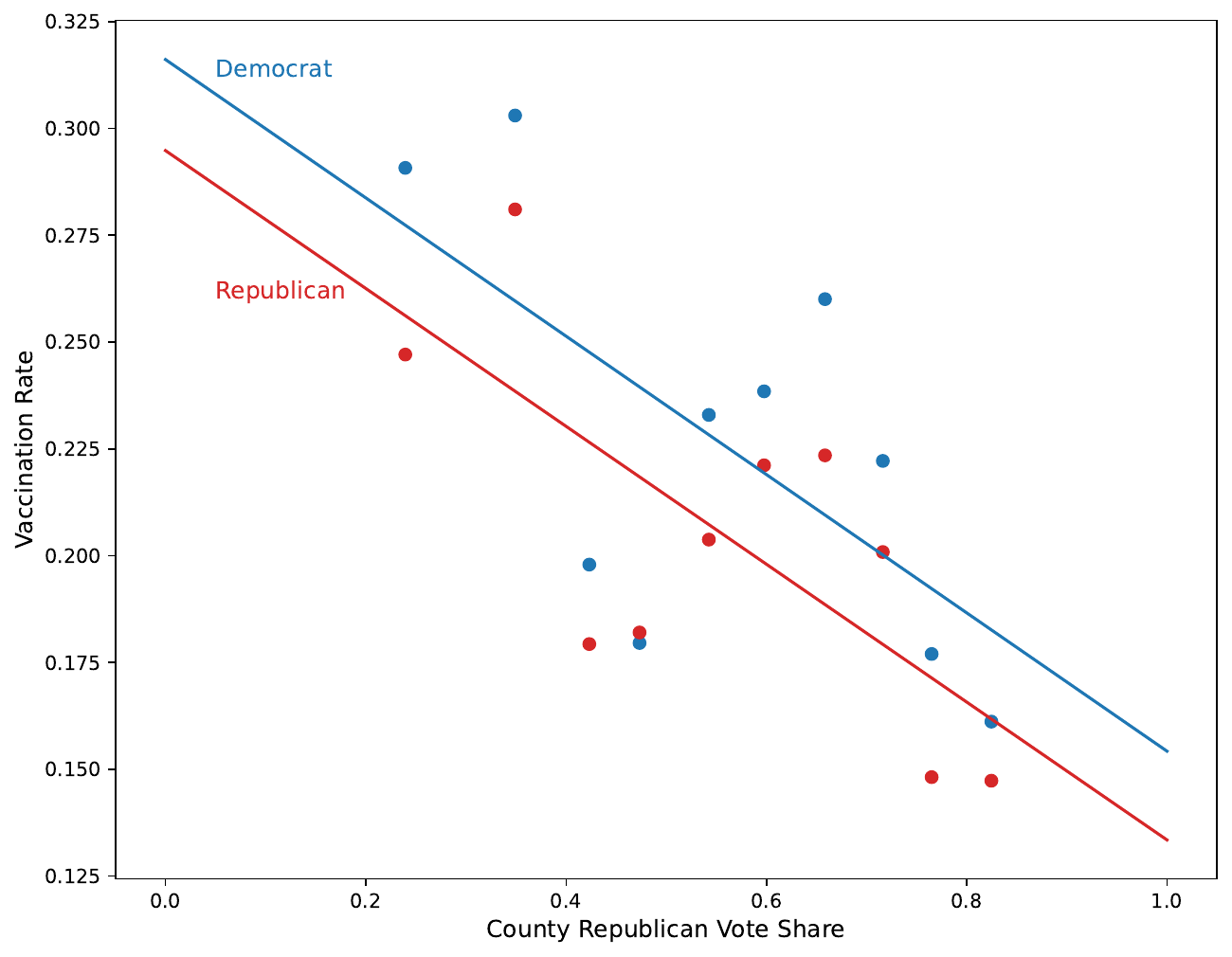}
        \vspace{-5pt} \\
    \end{centering}
{\footnotesize{Notes: The figure uses the matched auxiliary data set to report Covid-19 vaccination rates by county-level Republican vote share for individuals registered as Republicans (red) and individuals registered as Democrats (blue). Individuals are grouped into ten equal-sized bins based on the Republican vote share for the county in which they live. The average vertical difference between the blue and red points reflects the estimated sign of $\delta_B$. The average slope of the linear best fit lines reflects the estimated sign of $\delta_W$.  }}
    \label{fig:context_reinf}
\end{figure}

When contextual reinforcement holds, Corollary \ref{lem:sign_D_contextual_reinf} establishes that we can identify the sign of the difference in group means based on the sign of the difference between the ecological regression and neighborhood model estimators. As shown in Figure \ref{fig:party_vaccination_bounds}, we find that this difference is positive ($p<.01$), implying that Democrats are vaccinated at higher rates than Republicans. In turn, the contextual reinforcement bounds from Proposition \ref{prop:bounds_cr} imply that the partisan vaccination gap is between -47.9 and -5.5 percentage points (95\% CI: -51.8 to -5.1).

Figure \ref{fig:party_vaccination_bounds} summarizes our results under various identifying assumptions.\footnote{Figure \ref{fig:bounds_by_group} provides corresponding results for mean vaccination rates among Democrats and Republicans.} If $\delta_W$ or $\delta_B$ was assumed to be zero, the partisan vaccination gap would be point-identified at -47.9 or -5.5 percentage points, respectively. Conversely, without any assumptions beyond the observed data, the method of bounds interval for the partisan vaccination gap ranges from -77.4 to 52.8 percentage points. Turning to monotone EI, as discussed above there are plausible reasons to expect $\delta_B\leq0$ as well as $\delta_W\leq0$. Imposing these sign restrictions individually respectively tightens the method of bounds interval by 45\% and 22\%. Finally, under contextual reinforcement, our preferred identifying assumption, we can conclude that the partisan vaccination gap is between -47.9 and -5.5 percentage points, an interval that is 67\% smaller than the one obtained from the method of bounds. 

\begin{figure}[ht]      \caption{Identification of Partisan Vaccination Gap Using Monotone EI}   \label{fig:party_vaccination_bounds}
    \begin{centering} 
\includegraphics[width=0.9\textwidth]{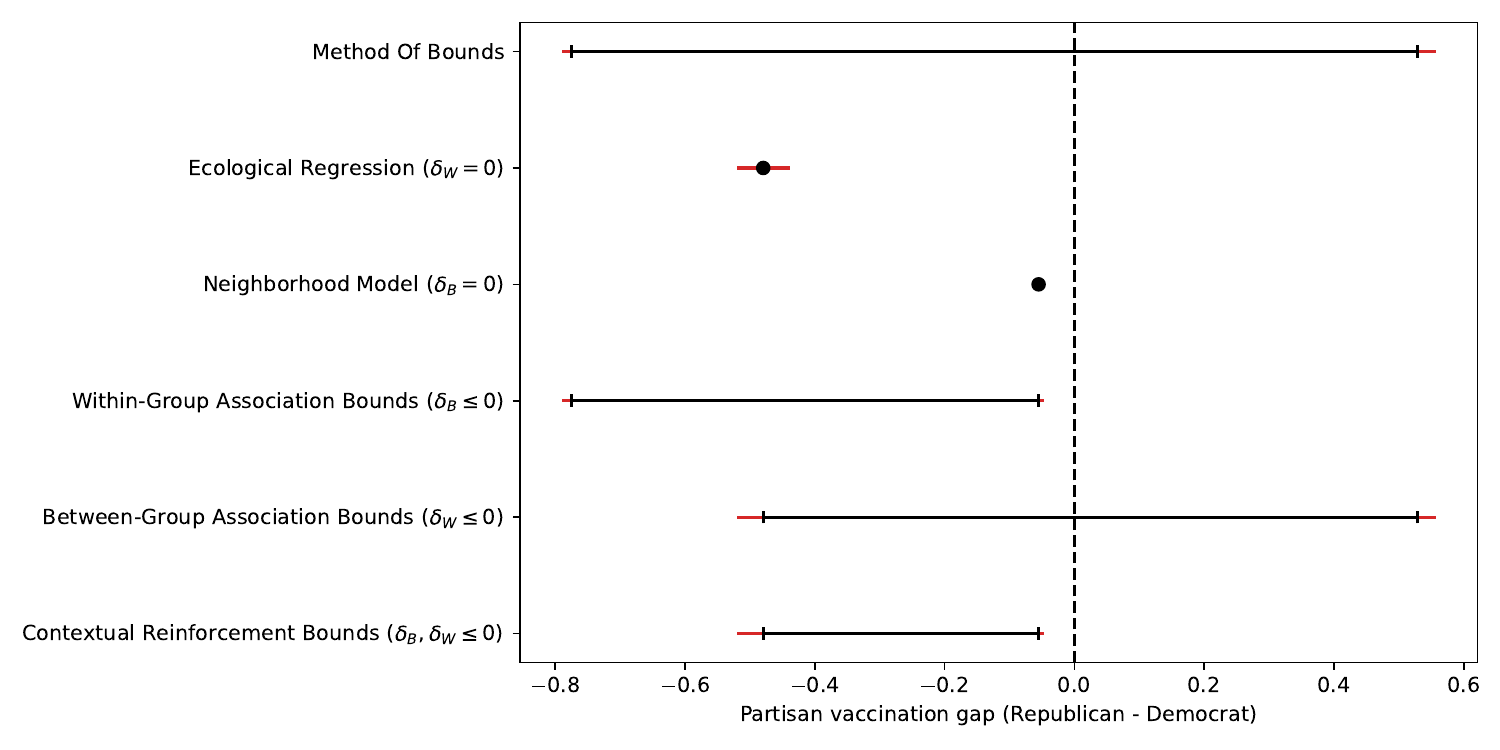}
    \vspace{-5pt} \\
    \end{centering}
{\footnotesize{Notes: The partisan vaccination gap is defined as the proportion of vaccinated Democrats subtracted from the proportion of vaccinated Republicans. A negative gap indicates that a higher proportion of Democrats than Republicans are vaccinated. Red bars are 95\% confidence intervals following \citep{imbens2004confidence}; the confidence intervals are based on standard errors from a county-level bootstrap with 1000 bootstrap replicates.}}
\end{figure}

We also demonstrate how our method may be used to more precisely bound county-specific vaccine disparities. For purposes of this exercise, we assume that the sign assumptions discussed above hold not only for $\delta_B$ and $\delta_W$ but locally for each county as well. To calculate the implied bounds, we estimate $\mu'(X_n)$ using population-weighted local linear approximation with an Epanechnikov kernel, with bandwidth chosen through 10-fold cross validation (Figure \ref{fig:local_linear}). Based on this estimated derivative, we use Proposition \ref{prop:cr_local}(ii) to calculate county-level bounds on the Democrat and Republican vaccination rates. For Contra Costa County, which leans left-of-center politically, we estimate that the Republican vaccination rate is between 0.50 and 0.77 and that the Democrat vaccination rate is between 0.77 and 0.87. For Galveston county, which leans right-of-center politically, the respective bounds for Republicans and Democrats are from 0.36 to 0.57 and from 0.57 to 0.89. Each of these intervals is substantially narrower than the corresponding interval derived from the method of bounds (see Figure \ref{fig:individual_disparities}).

\begin{figure}[ht]     
\caption{County-Specific Bounds}   
\label{fig:individual_disparities}
\begin{centering} 
\includegraphics[width=0.8\textwidth]{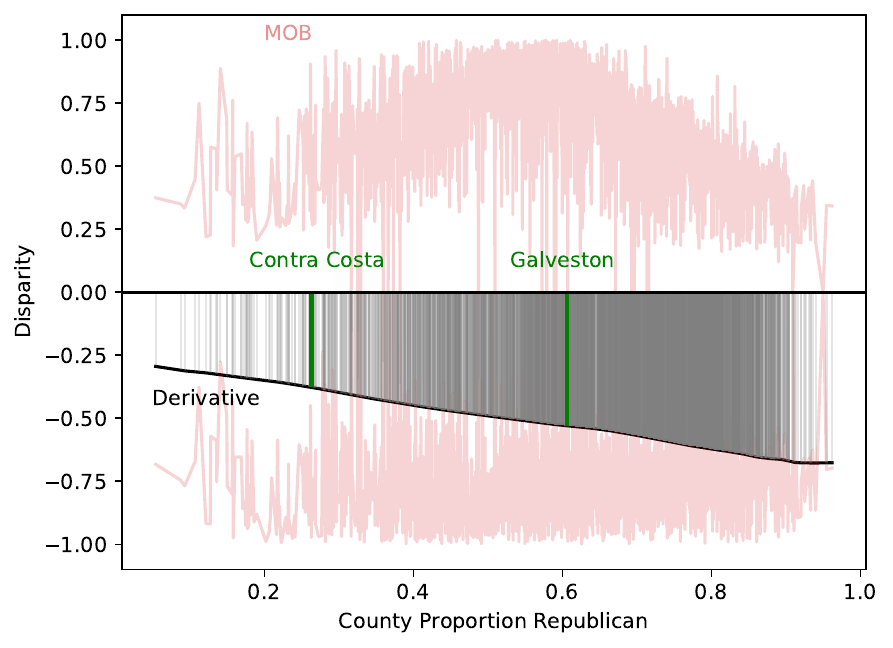}
    \vspace{-5pt} \\
    \end{centering}
{\footnotesize{Notes: This figure shows the Method of Bounds for the partisan vaccination gap for each county (red). The county-level bounds, based on the assumptions in Proposition \ref{prop:cr_local}, are shown in gray, and the estimated derivative of the conditional expectation function of the vaccination rate by county partisan makeup is shown in black. The implied bounds for Contra Costa County, California, and Galveston County, Texas, are highlighted in green.}}
\end{figure}

\section{Conclusion}

We study the classic statistical challenge of ecological inference (EI). Our results clarify the nature of the biases associated with canonical methods for point identification in EI settings. We use those results to derive novel partial identification results based on assumptions about the sign of the conditional associations between the outcome of interest and group membership or neighborhood. Although our approach requires additional structure relative to assumption-free tools like the method of bounds, the payoff to that additional structure can be substantially tighter bounds for the parameter of interest. In our empirical application, we illustrate how one can reason about the sign of the relevant conditional associations based on institutional knowledge and/or auxiliary individual-level data. Under plausible assumptions, the county-level data we rely on allows us to conclude that the Covid-19 vaccination rate among Republicans is between 5.5 and 47.9 (95\% CI: 5.1 and 51.8) percentage points lower than the corresponding rate among Democrats.

\bibliographystyle{apsr}
\bibliography{references.bib}

\appendix

\section*{Appendix: Proofs of Propositions}
\label{sec:app_proofs}

\begin{proof}[Proof of \textbf{Lemma \ref{lemma:nm_er}}]
Let $\gamma = \frac{Var(X_N)}{Var(X)}$. Then:

\begin{align*}
D_{NM} &= \frac{\mathbb{E}[X_NY_N]}{\mathbb{E}[X_N]}- \frac{\mathbb{E}[(1-X_N)Y_N]}{\mathbb{E}[1-X_N]}\\
&=\frac{\mathbb{E}[X_NY_N]-\mathbb{E}[X_N]\mathbb{E}[X_NY_N]-\mathbb{E}[Y_N]\mathbb{E}[X_N]+\mathbb{E}[X_N]\mathbb{E}[X_NY_N]}{\mathbb{E}[X_N](1-\mathbb{E}[X_N])}\\
&=\frac{\mathbb{E}[X_NY_N]-\mathbb{E}[Y_N]\mathbb{E}[X_N]}{\mathbb{E}[X_N](1-\mathbb{E}[X_N])}\\
&=\frac{Cov(X_N,Y_N)}{\mathbb{E}[X_N](1-\mathbb{E}[X_N])}\\
&=\frac{Var(X_N)}{\mathbb{E}[X_N](1-\mathbb{E}[X_N])} \frac{Cov(X_N,Y_N)}{Var(X_N)}\\
&= \frac{Var(X_N)}{\mathbb{E}[X](1-\mathbb{E}[X])} D_{ER} \\
&= \frac{Var(X_N)}{Var(X)} D_{ER} \\
&= \gamma D_{ER}\\
\end{align*}

as required. 

\end{proof}

We begin by formalizing an argument that we will repeatedly invoke in our proofs when we claim that identified bounds are \emph{sharp}:
\begin{lemma}[Feasibility of Intermediates]\label{lem:feasb_intermediates}
Suppose that we have two feasible marginal distributions $(Y^1,Y^0)=(A^1,A^0)$ and $(Y^1,Y^0)=(B^1,B^0)$, where feasible means that $Y_n^j \in [\underline{Y},\overline{Y}]$ for $j$ in $\{0,1\}$ and $Y_n = X_n Y_n^1+(1-X_n)Y_n^0$. Define also the respective disparities $D_A = A^1-A^0$ and $D_B=B^1-B^0$. Then for any $\alpha \in [0,1]$, the pair of marginals $(Y^1,Y^0) = (\alpha A^1+(1-\alpha)B^1,\alpha A^0+(1-\alpha)B^0)$ is feasible in the same sense and has disparity $\alpha D_A + (1-\alpha) D_B$. 
\end{lemma}
\begin{proof}
    To show feasibility, we need that for each $n \in \mathcal{N}$: 
    \begin{align*}
        X_n Y_n^1 + (1-X_n) Y_n^0 = Y_n\\
        Y_n^0, Y_n^1 \in [\underline{Y},\overline{Y}]
    \end{align*}

    So let: 
    \begin{align*}
        Y^1 = \alpha A^1 +(1-\alpha) B^1, \\ 
        Y^0 = \alpha A^0 +(1-\alpha)B^0.
    \end{align*}
    We designate the neighborhood-level distributions by  $Y_{n}^{1A}$ and $Y_{n}^{0A}$ for $(A^1,A^0)$ and $Y_{n}^{1B}$ and $Y_{n}^{0B}$ for $(B^1,B^0)$. Define 
    \begin{align*}
        Y_n^{1\alpha} = \alpha Y_n^{1A}+(1-\alpha) Y_{n}^{1B} \\
        Y_n^{0\alpha} = \alpha  Y_n^{0A}+(1-\alpha)Y_n^{0B}.
    \end{align*}
    We must have that $Y_n^{1\alpha},Y_n^{0\alpha} \in [\underline{Y},\overline{Y}]$, since $Y_n^{1\alpha}$ sits between $Y_n^{1A}$ and $Y_n^{1B}$, and similarly for $Y_n^{0\alpha}$, and we have that 
    \begin{align*}
        X_n Y_n^{1\alpha} +(1-X_n)Y_n^{0\alpha} &= X_n (\alpha Y_n^{1A}+(1-\alpha)Y_n^{1B})+(1-X_n)(\alpha Y_n^{0A}+(1-\alpha)Y_{n^{0B}}) \\&=\alpha (X_n Y_n^{1A}+(1-X_n) Y_n^{0A}) + (1-\alpha) (X_n Y_n^{1B}+(1-\alpha) Y_n^{0B})\\& = \alpha Y_n + (1-\alpha)Y_n\\&=Y_n.
    \end{align*}
    Thus, the choice of $Y_n^{1\alpha},Y_n^{0\alpha}$ is feasible. 

    For disparity, note that:
    \begin{align*}
        Y^{1\alpha}-Y^{0\alpha} &= \alpha A^1 + (1-\alpha) A^0 - (\alpha B^1  +(1-\alpha) B^0)
        \\&= \alpha D^A + (1-\alpha) D^B.
    \end{align*}
\end{proof}
\begin{corr}\label{lem:feasb_intermediate_disp}
    Suppose that $\underline{D},\overline{D}$ are two feasible bounds on disparity. Then for any $D \in [\underline{D},\overline{D}]$, $D$ is feasible as well, and in particular if $D = \alpha \underline{D}+(1-\alpha)\overline{D}$ for $\alpha \in [0,1]$, then a feasible distribution which gives rise to $D$ is the $\alpha$-mixture over the respective (possibly non-unique) distributions which give rise to $\underline{D}$ and $\overline{D}$. 
\end{corr}
\begin{proof}
 For any $D \in [\underline{D},\overline{D}]$ there is some $\alpha$ in $[0,1]$ such that $D=\alpha \underline{D}+(1-\alpha)\overline{D}$. Call the pair $(Y^1,Y^0) = (Y^{1A},Y^{0A})$ corresponding to $\underline{D}$ and $(Y^1,Y^0) = (Y^{1B},Y^{0B})$ for $\overline{D}$. By Lemma \ref{lem:feasb_intermediates}, the choice   $(Y^{1\alpha},Y^{0\alpha})=\alpha Y^{1A}+(1-\alpha) Y^{1B}, \alpha Y^{0A}+(1-\alpha)Y^{0B}$ is feasible and results in disparity $D$. Hence, $D$ is feasible. 
\end{proof}

\begin{lemma}[Neighborhood Values Characterized by the Method of Bounds]\label{lemma:gen_mob_char} 
For any neighborhood $n \in \mathcal{N}$ with  $\Pr[X=1|N=n]=X_n$ and $\underline{Y}\leq \E[Y|N=n]=Y_n \leq \overline{Y}$, the following are the most extreme values that $Y_n^1$ and $Y_n^0$ could take:
    \begin{align*}
    Y_{MOB,n}^{1+} :&=  \min\left\{\frac{Y_n-\underline{Y}(1-X_n)}{X_n},\overline{Y}\right\}, \\
    Y_{MOB,n}^{0-}:&=
    \max\left\{\frac{Y_n-\overline{Y} X_n}{1-X_n},\underline{Y}\right\},\\
    Y_{MOB,n}^{1-} :&= \max\left\{\frac{Y_n-\overline{Y}(1-X_n)}{X_n}, \underline{Y}\right\} , \\
    Y_{MOB,n}^{0+} :&=  \min\left\{\frac{Y_n-\underline{Y}X_n}{1-X_n},\overline{Y}\right\} .
\end{align*}
\end{lemma}
\begin{proof}
In any feasible distribution, we must have that:
\begin{align*}
    Y_n = X_n \cdot Y_n^1 + (1-X_n)\cdot Y_n^0,
\end{align*}
which can be rearranged to say that:
\begin{align*}
    Y_n^0 = \frac{Y_n-X_n\cdot Y_n^1}{1-X_n}
\end{align*} and
\begin{align*}
    Y_n^1 = \frac{Y_n-(1-X_n)\cdot Y_n^0}{X_n};
\end{align*}
and that
\begin{align*}
    Y_n^1, Y_n^0 \in [\underline{Y},\overline{Y}].
\end{align*}
Disparity is \emph{maximized} when as much of $Y_n$ as possible is coming from Group 1; it is \emph{minimized} when as much as possible is coming from Group 0. Focus on maximization for now. In addition to ensuring $Y_n^1 \in[\underline{Y},\overline{Y}]$, the maximal allocation to Group 1 must also leave $Y_n^0 \in [\underline{Y},\overline{Y}]$; that means that we must have
\begin{align*}
    Y_n^1 \geq \frac{Y_n-\overline{Y}(1-X_n)}{X_n} ;
    Y_n^1 \leq \frac{Y_n-\underline{Y}(1-X_n)}{X_n}; Y_n^1\leq \overline{Y};Y_n^1\geq\underline{Y}
\end{align*}
Notice that:
\begin{align*}
    \frac{Y_n-\underline{Y}(1-X_n)}{X_n}\geq \underline{Y} \iff Y_n - \underline{Y}(1-X_n)\geq \underline{Y}X_n \iff Y_n \geq \underline{Y}
\end{align*} which is always true by assumption, so we can eliminate the redundant inequality which is equivalent to  $Y_n^1 \geq \underline{Y}$; on the other hand, 
\begin{align*}
    \frac{Y_n-\underline{Y}(1-X_n)}{X_n} \leq  \overline{Y} \iff Y_n \leq \overline{Y}X_n+\underline{Y}(1-X_n),
\end{align*} which may or may not be true depending on where $Y_n$ falls, so the upper constraints are not redundant. Since the MOB then corresponds to the largest $Y_n^1$ that satisfies all these inequalities, we can write that
\begin{align*}
    Y_{MOB,n}^{1+} = \max\left\{\frac{Y_n-\overline{Y}(1-X_n)}{X_n}, \min\left\{\frac{Y_n-\underline{Y}(1-X_n)}{X_n},\overline{Y}\right\}\right\}.
\end{align*}
Notice that this can be further simplifed because
\begin{align*}
    \frac{Y_n -\overline{Y}(1-X_n)}{X_n} \leq \frac{Y_n-\underline{Y}(1-X_n)}{X_n},
\end{align*}
and
\begin{align*}
    \frac{Y_n-\overline{Y}(1-X_n)}{X_n} \leq \frac{\overline{Y} -\overline{Y}(1-X_n)}{X_n} = \frac{\overline{Y} X_n}{X_n}=\overline{Y},
\end{align*}
so the minimum of those two is always larger than the first argument. Then:
\begin{align*}
    Y_{MOB,n}^{1+} = \min\left\{\frac{Y_n-\underline{Y}(1-X_n)}{X_n},\overline{Y}\right\}.
\end{align*}

$Y_{MOB,n}^{0-}$ is the remainder, i.e.:

\begin{align*}
    Y_{MOB,n}^{0-}:&=\frac{Y_n-Y_{MOB,n}^{1+}X_n}{1-X_n}\\
    &= \frac{Y_n- \min\left\{Y_n-\underline{Y}(1-X_n),X_n\overline{Y}\right\}}{1-X_n}\\
    &= \frac{\max\left\{ \underline{Y}(1-X_n), Y_n-X_n \overline{Y} \right\} }{1-X_n}\\
    &= \max\left\{ \underline{Y}, \frac{Y_n - X_n \overline{Y}}{1-X_n}\right\}
\end{align*}

which the form in which the result is stated.

Turning to minimization of disparity, the logic is much the same, but in the opposite direction. We must have that:
\begin{align*}
     Y_n^0 \geq \frac{Y_n-\overline{Y}X_n}{1-X_n}; Y_n^0 \leq \frac{Y_n -\underline{Y}X_n}{1-X_n};Y_{n}^0 \leq \overline{Y}; Y_{n}^0\geq \underline{Y};
\end{align*}
again removing the first redundant constraints gives that:
\begin{align*}
    Y_{MOB,n}^{0+} := \max\left\{\frac{Y_n-\overline{Y}X_n}{1-X_n}, \min\left\{\frac{Y_n-\underline{Y}X_n}{1-X_n},\overline{Y}\right\}  \right\} = \min\left\{\frac{Y_n-\underline{Y}X_n}{1-X_n},\overline{Y}\right\}  
\end{align*}
and:
\begin{align*}
    Y_{MOB,n}^{1-} :&= \frac{Y_n-Y_{MOB,n}^{0+}(1-X_n)}{X_n}
    \\&= \frac{Y_n-\min\left\{\frac{Y_n-\underline{Y}X_n}{1-X_n},\overline{Y}\right\} (1-X_n)}{X_n} \\
    &= \frac{Y_n-\min\left\{Y_n-\underline{Y}X_n,\overline{Y}  (1-X_n)\right\}}{X_n}\\
    &= \frac{\max \left\{\underline{Y}X_n,Y_n-\overline{Y}  (1-X_n)\right\}}{X_n}\\
    &= \max \left\{\underline{Y},\frac{Y_n-\overline{Y}  (1-X_n)}{X_n}\right\}
\end{align*}

where we have again distributed the max and min to simplify the expression, re-arranged terms for consistency, and eliminated the redundant constraint.
\end{proof}

Now we turn to proving Proposition 1. 
\begin{proof}[Proof of Proposition \ref{prop:mob_levels}]
Note that Lemma \ref{lemma:gen_mob_char} gives the neighborhood-level characterization of the most extreme values of $Y_n^1$ and $Y_n^0$ that are still consistent with the data and bounds. 
First, we characterize the values that can be achieved. Define 
    \begin{align*}
        Y_{MOB}^{g\sigma} = \E[Y_{MOB,N}^{g\sigma}]
    \end{align*}
    for $g \in\{0,1\}$ and $\sigma \in \{+,-\}$. This corresponds to the expected value of $Y$ for group $G$ under the neighborhood-level extremes obtained via the method of bounds as described in Lemma \ref{lemma:gen_mob_char}.
    We can write that:
    \begin{align*}
        \E[Y_{MOB,N}^{g\sigma}] = \sum_{n\in \mathcal{N}} \Pr[N=n|X=g]Y_{MOB,n}^{g\sigma}
    \end{align*}
    Then noting that by Bayes' Rule:
    \begin{align*}
    \Pr[N=n|X=g] = \frac{\Pr[X=g|N=n] \Pr[N=n]}{\Pr[X=g]}
    \end{align*}
    which is $X_n p_n/\E[X]$ for $X=1$ and $(1-X_n)p_n/(1-\E[X])$ for $X=0$. So we can write:
    \begin{align*}
        Y_{MOB}^{1+} &= \sum_{n \in \mathcal{N}}\frac{X_n p_n}{\E[X]}\cdot  \min\left\{\frac{Y_n-\underline{Y}(1-X_n)}{X_n},\overline{Y}\right\}
        \\&=\frac{1}{\E[X]}\sum_{n \in \mathcal{N}} p_n \min\left\{Y_n-\underline{Y}(1-X_n),\overline{Y}X_n\right\}
        \\& = 
        \frac{\E\left[\min\left\{Y_n-\underline{Y}(1-X_n),\overline{Y}X_n\right\}\right]}{\E[X]}.
    \end{align*}
    Similarly,
    \begin{align*}
        Y_{MOB}^{0-}&=\sum_{n \in \mathcal{N}}\frac{p_n (1-X_n)}{1-\E[X]} \max\left\{ \underline{Y}, \frac{Y_n - X_n \overline{Y}}{1-X_n}\right\}
    \\&=\frac{\E[\max\left\{ \underline{Y}(1-X_n), Y_n - X_n \overline{Y}\right\}]}{1-\E[X]}
    \end{align*}

    The remaining two expressions follow similarly.

Now we turn to sharpness. First recall that by Lemma \ref{lemma:gen_mob_char}, we know that for each $n\in\mathcal{N}$, $Y_{MOB,n}^{1+}$ is the maximum that $Y_{n}^{1}$ could be; since each neighborhood is independent, the maximum that $Y^1$ could be is obtained by maximizing $Y_n^1$ for each $n$, and so $\E[Y_{MOB,N}^{1+}]=Y_{MOB}^{1+}$ is indeed the maximum that $Y^1$ could be. Since maximizing the value of $\E[Y|N=n,X=1]$ is tantamount to minimizing $\E[Y|N=n,X=0]$, this also shows that $Y^0$ cannot be less than $\E[Y^{0-}_{MOB,N}] = Y_{MOB}^{0-}$. If $Y_n=Y_{MOB,n}^{0-}$ for all $n$, then $Y^0=Y_{MOB}^{0-}$. Similarly, $Y^0$ can be at most $Y_{MOB}^{0+}$ and $Y^1$ must be at least $Y_{MOB}^{1-}$. We thus have that:
\begin{align*}
    Y^1 &\in [Y_{MOB}^{1-},Y_{MOB}^{1+}]\\
    Y^0 &\in [Y_{MOB}^{0-}, Y_{MOB}^{0+}]
\end{align*}
To see that these bounds are sharp - i.e. that there exists some joint distribution consistent with them - simply note the feasibility of $(Y_{MOB,n}^{1+},Y_{MOB,n}^{0-})$ and $(Y_{MOB,n}^{1-},Y_{MOB,n}^{0+})$ for all $n \in \mathcal{N}$ and shows that $(Y^1,Y^0) = (Y_{MOB}^{1+},Y_{MOB}^{0-})$   and $(Y^1,Y^0) = (Y_{MOB}^{1-},Y_{MOB}^{0+})$  are feasible. Then, for any $Y^1 \in [Y_{MOB}^{1-},Y_{MOB}^{1+}]$, pick $\alpha$ so that $Y^1 = \alpha Y_{MOB}^{1+}+(1-\alpha) Y_{MOB}^{1-}$. Lemma \ref{lem:feasb_intermediates} shows that there exists a marginal distribution that makes this feasible - in particular, $(Y_{n}^{1},Y_n^0 )=\left( \alpha Y_{MOB,n}^{1+}+(1-\alpha)Y_{MOB,n}^{1-},\alpha Y_{MOB,n}^{0-}+(1-\alpha)Y_{MOB,n}^{0+}\right)$.  The same argument can be made for $Y^0 \in [Y_{MOB}^{0-}, Y_{MOB}^{0+}]$. Thus, the bounds are sharp.  

\end{proof}

\begin{proof}[Proof of \textbf{Proposition \ref{prop:mob_dif}}]
The formulas for $D_{MOB}^+$ and $D_{MOB}^-$ follow by direct calculation from Proposition \ref{prop:mob_levels}; sharpness follows from Corollary \ref{lem:feasb_intermediate_disp} and Proposition \ref{prop:mob_levels}. We show the calculation for $D_{MOB}^+$ below; the calculation for $D_{MOB}^-$ is similar. 
\begin{align*}
    D_{MOB}^+ &=Y_{MOB}^{1+}-Y_{MOB}^{0-}
    \\ &= \frac{\E[\min\left\{Y_n-\underline{Y}(1-X_n),\overline{Y}X_n\right\}]}{\E[X]} - \frac{\E[\max\left\{\underline{Y}(1-X_n),Y_n-\overline{Y}X_n\right\}]}{1-\E[X]}\\
    &= \frac{\E[\min\left\{Y_n-\underline{Y}(1-X_n),\overline{Y}X_n\right\}]}{\E[X]} - \frac{\E[Y_n - \min\left\{Y_n-\underline{Y}(1-X_n),\overline{Y}X_n\right\}]}{1-\E[X]}\\
    &= \frac{\E[\min\left\{Y_n-\underline{Y}(1-X_n),\overline{Y}X_n\right\}]}{\E[X](1-\E[X])}-\frac{\E[Y_n]}{1-\E[X]}\\
    &= \frac{\E[\min\left\{Y_n-\underline{Y}(1-X_n),\overline{Y}X_n\right\}]-\E[Y]\E[X]}{\E[X](1-\E[X])}\\
    &= \frac{\E[\min\left\{Y_n-\underline{Y}(1-X_n),\overline{Y}X_n\right\}]-\E[Y]\E[X]}{\Var(X)}\\
\end{align*}
Where the third equality follows from the fact that $A - \min\{B,C\} = \max \{A-B,A-c\}$, the fifth from the fact that $\E[Y_n]=\E[Y]$, and the final from the fact that $ \Var(X) = \E[X](1-\E[X])$ for any binary random variable $X$.

\end{proof}

\begin{proof}[Proof of \textbf{Proposition \ref{prop:decomp_conditional_assoc}}]
    We have $D_{NM} = \gamma D_{ER}$. Substituting our expressions in propositions \ref{prop:bias_er} and \ref{prop:bias_nm}, we get
\begin{align*}
    D-\frac{\delta_B}{\Var(X)} &= \gamma (D+\frac{\delta_W}{\Var(X_N)})\\
    D(1-\gamma) &= \frac{\delta_B}{\Var(X)} + \frac{\gamma\delta_W}{\Var(X_N)}\\
    D(1-\gamma) &= \frac{\delta_B}{\Var(X)} + \frac{\delta_W}{\Var(X)}\\
    D &= \frac{\delta_B+\delta_W}{(1-\gamma)\Var(X)}
\end{align*}
    as required.
\end{proof}

Before proving Proposition \ref{prop:bias_er}, we first establish a technical lemma:

\begin{lemma}\label{lem:cov_equality}
       \begin{align*} \Cov(\mathbb{E}[Y|X],\mathbb{E}[X_N|X]) = D\cdot \Var[X_N]
       \end{align*}

\end{lemma}
\begin{proof}
         First note that:
    \begin{align*}
    \mathbb{E}[X_N|X=1] = \sum_{n \in \mathcal{N}} \Pr[N=n|X=1] \cdot X_n = \sum_{n \in \mathcal{N}} \frac{X_n \cdot p_n}{\Pr[X=1]} X_n = \frac{\mathbb{E}[X_N^2]}{\mathbb{E}[X]} 
    \end{align*}
    where the first equality follows by the definition of conditional expectation and the second by Bayes' Rule. 
    Similarly,
    \begin{align*}
        \mathbb{E}[X_N|X=0] = \sum_{n\in\mathcal{N}} \Pr[N=n|X=0] \cdot X_n = \sum_{n\in\mathcal{N}} \frac{(1-X_n)p_n}{\Pr[X=0]}\cdot X_n = \frac{\mathbb{E}[X]-\mathbb{E}[X_N^2]}{1-\mathbb{E}[X]}
    \end{align*}

    Now, turning to the main issue:
    \begin{align*}
    \Cov(\mathbb{E}[Y|X],\mathbb{E}[X_N|X])& = \E\left[\mathbb{E}[Y|X]\cdot \mathbb{E}[X_N|X]] - \mathbb{E}[\mathbb{E}[Y|X]\right] \cdot \E \left[\mathbb{E}[X_N|X]\right]
    \\& = \E\left[\mathbb{E}[Y|X]\cdot \mathbb{E}[X_N|X]\right] - \mathbb{E}[Y] \cdot \mathbb{E}[X]
    \end{align*}
    where the second equality follows by the law of iterated expectations.
    Now note that:
    \begin{align*}
        \E\left[\mathbb{E}[Y|X]\cdot \mathbb{E}[X_N|X]\right] 
        &=\Pr[X=1]\cdot\left(\mathbb{E}[Y|X=1]\cdot \mathbb{E}[X_N|X=1]\right)\\& +(1-\Pr[X=1])\cdot \left(\mathbb{E}[Y|X=0]\cdot \mathbb{E}[X_N|X=0]\right)
        \\&= \mathbb{E}[X] \cdot \left(\mathbb{E}[Y|X=1] \frac{\mathbb{E}[X_N^2]}{\mathbb{E}[X]}\right) + (1-\mathbb{E}[X])\cdot \left(\mathbb{E}[Y|X=0] \cdot \frac{\mathbb{E}[X]-\mathbb{E}[X_N^2]}{1-\mathbb{E}[X]}\right)
        \\&= \mathbb{E}[Y|X=1] \mathbb{E}[X_N^2] +E[Y|X=0] \cdot \mathbb{E}[X]-\mathbb{E}[X_N^2]\cdot \mathbb{E}[Y|X=0]
    \end{align*}
    Now collecting terms, we have this equals:
    \begin{align*}
        \mathbb{E}[X_N^2]\left(\mathbb{E}[Y|X=1]-\mathbb{E}[Y|X=0]\right) + \mathbb{E}[Y|X=0]\cdot \mathbb{E}[X]= \mathbb{E}[X_N^2]\cdot D + \mathbb{E}[X]\cdot \mathbb{E}[Y|X=0]
    \end{align*}

Substituting this result back into $\Cov(E[Y|X],\mathbb{E}[X_N|X])$, we have:
\begin{align*}
    \Cov(\mathbb{E}[Y|X],\mathbb{E}[X_N|X]) &= \mathbb{E}[X_N^2]\cdot D + \mathbb{E}[X]\cdot \mathbb{E}[Y|X=0]- \mathbb{E}[X]\cdot \mathbb{E}[Y]
    \\ & = \mathbb{E}[X_N^2] \cdot D + \mathbb{E}[X]\cdot \left(\mathbb{E}[Y|X=0]-\mathbb{E}[Y]\right)
\end{align*}
But note that:
\begin{align*}
    \mathbb{E}[Y|X=0]-\mathbb{E}[Y] &= \mathbb{E}[Y|X=0] - \mathbb{E}[X]\cdot \mathbb{E}[Y|X=1]-(1-\mathbb{E}[X])\cdot\mathbb{E}[Y|X=0]
    \\&= \mathbb{E}[X]\cdot\mathbb{E}[Y|X=0]-\mathbb{E}[X]\cdot \mathbb{E}[Y|X=1]
    \\&=-\mathbb{E}[X]\cdot D
\end{align*}
Substituting this back into the previous equation, we have:
\begin{align*}
    \Cov(\mathbb{E}[Y|X],\mathbb{E}[X_N|X]) &= \mathbb{E}[X_N^2] \cdot D + \mathbb{E}[X]\cdot(-\mathbb{E}[X]\cdot D)
    \\&= D\cdot\left(\mathbb{E}[X_N^2] -\mathbb{E}[X]^2\right)
\end{align*}
Since $\mathbb{E}[X_N]=\mathbb{E}[X]$, the latter term is $\mathbb{E}[X_N^2]-\mathbb{E}[X_N]^2=\Var[X_N]$. Then we have that:
\begin{align*}
    \Cov(\mathbb{E}[Y|X],\mathbb{E}[X_N|X]) = D\cdot \Var(X_N)
\end{align*}

\end{proof}

\begin{proof}[Proof of \textbf{Proposition \ref{prop:bias_er}}]
    Recall that:
    \begin{align*}
        D_{ER} = \frac{\Cov(Y,X_N)}{\Var[X_N]}= \frac{\mathbb{E}[\Cov(Y,X_N|X)]+\Cov(\mathbb{E}[Y|X], \mathbb{E}[X_N|X])}{\Var[X_N]},
    \end{align*}
    where the last equality is by the law of total covariance.
    Since $\delta_{W}$ is defined as $\mathbb{E}[\Cov(Y,X_N|X)]$, we have already shown that:
    \begin{align*}
        D_{ER} = \frac{\Cov(\mathbb{E}[Y|X],\mathbb{E}[X_N|X]) + \delta_{W}}{\Var[X_N]}. 
    \end{align*}
    It thus suffices to show that: $\Cov(\mathbb{E}[Y|X],\mathbb{E}[X_N|X]) = D\cdot \Var[X_N]$. But that is exactly what is shown in Lemma \ref{lem:cov_equality}.

Thus:
\begin{align*}
    D_{ER} = \frac{\Cov(\mathbb{E}[Y|X],\mathbb{E}[X_N|X]) + \delta_{W}}{\Var(X_N)} = \frac{ D\cdot \Var(X_N)+\delta_{W}}{\Var(X_N)} = D+ \frac{\delta_{W}}{\Var(X_N)}
\end{align*}
as desired.

\end{proof}

\begin{proof}[Proof of \textbf{Proposition \ref{prop:cna_bounds}}]
Note that by definition of the method of bounds, $D \in [D_{MOB}^{-},D_{MOB}^{+}]$. 

\begin{itemize}
\item[(i)] Proposition \ref{prop:bias_er} implies that if $\delta_W \geq0$, $D\leq D_{ER}$. Since we always have that $D\leq D_{MOB}^{+}$ and  $D \leq D_{ER}$, it is less than their minimum; combining that with the fact that $D_{MOB}^{-} \leq D$ yields (i) .

\item[(ii)] follows similarly: By Proposition \ref{prop:bias_er}, $\delta_W \leq 0 \implies D_{ER}\leq D $, and since $D \geq D_{MOB}^{-}$, $D$ is thus larger than their maximum, and also $D_{MOB}^{+}\geq D$. 

\item[(iii)] 
By definition of the method of bounds, \emph{no} value of $D$ outside $[D_{MOB}^{-},D_{MOB}^{+}]$ can be attained by any distribution consistent with the information known. On the other hand, for each choice of $D$ in $[D_{MOB}^{-},D_{MOB}^{+}]$, Proposition \ref{prop:mob_dif} shows that there is a joint distribution resulting in $D$ that is consistent with $Y_n$ for all $n$. The question of whether each proposed range is sharp given the information about $\delta_W$ thus reduces to whether knowing the sign of $\delta_W$ allows us to eliminate any more possibilities for $D$ that are contained in the range. We prove sharpness of (i), in which $\delta_W\geq0$; the proof of (ii) is similar, mutatis mutandis.

\paragraph{Subcase (a)} Suppose $D_{ER}\geq D_{MOB}^+$. 
Then our bounds on $D$ are $[D_{MOB}^{-},D_{MOB}^{+}]$. By definition of the MOB, all values of $D$ outside this range are infeasible. We must now show that all values inside of this range must have $\delta_W \geq 0$.  

Choose any $D \in [D_{MOB}^{-},D_{MOB}^{+}]$. 
By Proposition \ref{prop:bias_er}, we know that 
\begin{align*}
    D_{ER}= D+ \frac{\delta_W}{\Var[X]} \implies \delta_W = (D_{ER}-D)\Var[X]
\end{align*}
Since $\Var[X]$ is always nonnegative, and $D_{ER} \geq D_{MOB}^{+} \geq D$, we must have $\delta_W \geq 0$, as required. 

\paragraph{Subcase (b)} Suppose $D_{ER}<D_{MOB}^+$. Then our bounds on $D$ are $[D_{MOB}^-, D_{ER}]$. By definition of the MOB, all values of $D$ outside of $[D_{MOB}^{-},D_{MOB}^{+}]$ are infeasible, as are values in $(D_{ER},D_{MOB}^{+}$ as show in part (i). We must now show that all the values inside $[D_{MOB}^{-},D_{ER}]$ are feasible.

So suppose $D' \in [D_{MOB}^{-},D_{ER}]$. Then we know that $D' \leq D_{ER}$, and then by the same logic as subcase a, $\delta_W\geq0$, as required. 

In both subcases, we have shown that all choices of $D'$ between $D_{MOB}^-$ and the smaller of $\min\left\{D_{ER},D_{MOB}^+\right\}$ are consistent with the known information and are feasible, and that these are the only such feasible choices. Thus, the bound is sharp.

\end{itemize}
\end{proof}

\begin{proof}[Proof of \textbf{Proposition \ref{prop:bias_nm}}]
    By the law of iterated expectations, we have that:
    \begin{align*}
        \Cov(X,Y) &= \mathbb{E}[\Cov(X,Y|N)]+\Cov(\mathbb{E}[X|N], \mathbb{E}[Y|N]) \\&= \mathbb{E}[\Cov(X,Y|N)] + \Cov(X_N,Y_N)
    \end{align*}
    Thus 
    \begin{align*}
        D_{NM}= \frac{\Cov(X_N,Y_N)}{\Var (X)}= \frac{\Cov(X,Y)-\mathbb{E}[\Cov(X,Y|N)]}{\Var(X)} = D - \frac{\delta_{B}}{\Var(X)}
    \end{align*}
as desired.
\end{proof}

Before proving Proposition \ref{prop:cga_bounds}, we first establish a technical lemma:

\begin{lemma}[Neighborhood model is always feasible]\label{lem:nm_feasb}
\begin{align*}
    D_{NM} \in [D_{MOB}^{-},D_{MOB}^{+}]
\end{align*}
\end{lemma}
\begin{proof}
    Let $Y_{n,NM}^1=Y_{n,NM}^0=Y_{n}$. Then $X_n\cdot Y_{n,NM}^1 + (1-X_n)\cdot Y_{n,NM}^0 = Y_{NM}$; in other words, the joint distribution corresponding to equal values of $E[Y|X,N=n]$, which is always possible, generates the neighborhood model's disparity. Since the MOB captures all feasible distributions, it captures the neighborhood model as well. 
\end{proof}

\begin{proof}[Proof of \textbf{Proposition \ref{prop:cga_bounds}}]
Note that by definition of the method of bounds, $D \in [D_{MOB}^{-},D_{MOB}^{+}]$. Also, by Lemma \ref{lem:nm_feasb}, $D_{NM}$ is feasible, and thus $D_{NM} \in [D_{MOB}^{-},D_{MOB}^{+}]$.

\begin{itemize}
\item[(i)] Proposition \ref{prop:bias_nm} implies that if $\delta_B \geq 0$, $D\geq D_{NM}$. Since $D \geq D_{NM} \geq D_{MOB}^{-}$, (i) follows. 

\item[(ii)] follows similarly: By Proposition \ref{prop:bias_nm}, $\delta_B \leq 0 \implies D_{NM}\geq D $, and since $D \leq D_{NM} \leq D_{MOB}^{+}$, (ii) follows.

\item[(iii)] Proof is similar to that  of Proposition \ref{prop:cna_bounds} (iii). All that is necessary to show is that the bounds cannot be further narrowed by finding points that conflict with information about $\delta_B$. For $i$, any $D'$ in $[D_{NM}, D_{MOB}^+]$, if it were to conflict with $\delta_B\geq0$, it would have to be that $D_{NM}+\frac{\delta_B}{\Var[X_n]} = D' \implies \delta_B <0$, but that would require $D'<D_{NM}$, a contradiction. A similar proof, mutatis mutandis, works for (ii).

\end{itemize}
\end{proof}

Before proving Proposition \ref{prop:bounds_cr}, we first establish a technical lemma:

\begin{lemma}[Bound Relationships]\label{lem:bound_relations}
It is always the case that:
\begin{align*}
    \sign(D_{ER})=\sign(D_{NM}).
\end{align*}
However: 
\begin{align*}
|D_{ER}|\geq |D_{NM}|
\end{align*}
\end{lemma}
\begin{proof}
    Recall that:
    \begin{align*}
        D_{NM}= \frac{\Cov(Y,X_N)}{\Var[X]} = \frac{\Cov(Y,X_N)}{\Var[X_N]}\cdot \frac{\Var[X_N]}{\Var[X]}=D_{ER}\cdot\frac{\Var[X_N]}{\Var[X]} .
    \end{align*}
    In other words, $D_{NM}$ is a scalar multiple of $D_{ER}$ (with a positive scalar), so they must share the same sign. On the other hand, $\Var[X_N]<\Var[X]$, so the scalar multiple is smaller than 1. 
\end{proof}

\begin{proof}[Proof of \textbf{Proposition \ref{prop:bounds_cr}}]
By Lemma \ref{lem:bound_relations}, $D_{ER}$ and $D_{NM}$ have the same sign, i.e. either both are positive or both are negative, yet $D_{ER}$ is always of larger magnitude that $D_{NM}$. In other words, either $0\leq D_{NM} \leq D_{ER}$ or $D_{ER} \leq D_{NM}\leq 0$. 

Thus, under (i), we must have that $0 \leq D_{NM} \leq D_{ER}$. By combining Proposition \ref{prop:bias_er} and Proposition \ref{prop:bias_nm} with the fact that $\delta_W\delta_B\geq0$, we must have that $D$ is in between $D_{NM}$ and $D_{ER}$. Thus $0 \leq D_{NM} \leq D \leq D_{ER}$. But since the MOB's upper estimate is an upper bound on disparity,  we also have $D \leq D_{MOB}^{+}$. Thus we can write that:
\begin{align*}
    0 \leq D_{NM} \leq D \leq \min\left\{D_{MOB}^{+}, D_{ER}\right\}.
\end{align*}

For (ii), similar reasoning around Lemma \ref{lem:bound_relations} and Propositions \ref{prop:bias_er} and \ref{prop:bias_nm}, mutatis mutandis, gives that $D_{ER} \leq D \leq 0$. Again, since the MOB lower bound must lower bound the disparity, we also must have that $D_{MOB}^{-} \leq D$. Thus we can write that:
\begin{align*}
    \max\left\{D_{MOB}^{-}, D_{ER}\right\} \leq D \leq D_{NM}\leq 0
\end{align*}

For (iii), we must prove that these bounds are sharp. Notice first of all that by Lemma \ref{lem:nm_feasb}, $D_{NM}$ is always feasible, so in (i) we know that $D_{NM} \geq D_{MOB}^{-}$; there is thus nothing that can be gained via the MOB for the lower end of the interval. A similar argument shows that nothing can be gained on the upper end of the interval via the MOB in (ii). Finally, note that any $D'$ within the interval of (i) can be mapped to some $\alpha \in [0,1]$ such that $D'$ is $\alpha \cdot D_{NM}+(1-\alpha)\cdot \min\left\{D_{MOB}^{+},D_{ER}\right\}$. Such a mixture will have $\mathbb{E}[\Cov(Y,X_N|X)]$, $\mathbb{E}[\Cov(Y,X|X_N)] \geq 0$ since both extremes do; thus this mixture does not contradict any of the given information and is possible. The bounds are thus sharp; a similar argument can be made for (ii).

\end{proof}

\begin{proof}[Proof of \textbf{Theorem \ref{thm:identification_D}}]
This Theorem is a consequence of the previous propositions, which show sharp bounds under the various possible conditions. First off, note that if $\delta_B,\delta_W \in\mathbb{R}$, then nothing is known besides what the method of bounds provide; thus Proposition \ref{prop:mob_dif} gives the top left cell. Now note that if $\delta_B=\delta_W=0$, then by Proposition \ref{prop:decomp_conditional_assoc}, $D=0$. If $\delta_W=0$, $D=D_{ER}$, which takes care of the bottom row, while $\delta_B=0\implies D=D_{NM}$ takes care of the rightmost column. The shaded cells are the results of Proposition \ref{prop:bounds_cr}. Proposition \ref{prop:cga_bounds} give the remaining two cells of the first row; Proposition \ref{prop:cna_bounds} gives the remaining two cells of the first column. For each of these cells, note that the we have already shown the bounds are sharp as in the respective propositions.

Two cells remain. For $\delta_B \geq0$, $\delta_W\leq 0$, Proposition \ref{prop:bias_er} shows that $D = D_{ER}-\frac{\delta_W}{\Var(X_N)}\geq D_{ER}$, while Proposition \ref{prop:bias_nm} shows that $D=D_{NM}+\frac{\delta_B}{\Var(X_N)}\geq D_{NM}$. Thus, $D \geq \max\left\{D_{ER},D_{NM}\right\}$; but we also have that $D \leq D_{MOB}^+$, and $D\geq D_{MOB}^-$ by the method of bounds. Thus $D \geq \max\left\{D_{ER},D_{NM},D_{MOB}^{-}\right\}$ and $D \leq D_{MOB}^{+}$; combining these yields the cell. To see that the bound is sharp, note first of all since the range is a subset of the method of bounds, Corollary \ref{lem:feasb_intermediate_disp} shows that any $D$ does not conflict with $Y_n$ because it is a subset of the method of bounds range. Thus to be infeasible given the stated information, $D$ would have to result in $\delta_B < 0$ or $\delta_W >0$; but neither of these can be true for the stated range (simply because we \emph{derived} the stated range by identifying the region where these could not be true). 
The bound is thus sharp. 

On the other hand, the same propositions show that $\delta_B \leq 0 \implies D\leq D_{NM}$, while $\delta_W \geq 0 \implies D \leq D_{ER}$. Again, we also have that $D_{MOB}^{+}\geq D$, so $\min\left\{D_{ER},D_{NM},D_{MOB}^{+}\right\}\geq D$, and again $D\geq D_{MOB}^{-}$. Combining these yields the cell, sharpness follows similarly, and we have thus proved each cell in the table is as claimed. 
\end{proof}

\begin{proof}[Proof of \textbf{Proposition \ref{prop:bias_er_levels}}]
Note that the claim for $Y_{ER}^1-Y^1$ will be satisfied if we can show that:
\begin{align*}
   \frac{\Var(X_N)}{1-\mathbb{E}[X]} \cdot \left(Y_{ER}^{1}-Y^1\right) = \delta_W
\end{align*}
Recall first off that:
\begin{align*}
    \gamma = \frac{\Var(X_N)}{\Var(X)}= \frac{\Var(X_N)}{\mathbb{E}[X](1-\mathbb{E}[X])}\implies \frac{\Var(X_N)}{1-\mathbb{E}[X]} = \gamma \mathbb{E}[X].
\end{align*}
Now, let us look first at $Y_{ER}^1$. We have:
\begin{align*}
    \gamma\mathbb{E}[X]\cdot  Y_{ER}^1&= \gamma\mathbb{E}[X]\cdot\left[\mathbb{E}[Y]+(1-\mathbb{E}[X])\frac{\Cov(X_N,Y_N)}{\Var(X_N)}\right]
    \\& = \gamma \mathbb{E}[X]\mathbb{E}[Y] + \Cov(X_N,Y_N)\cdot \frac{\mathbb{E}[X]\cdot\E(1-\mathbb{E}[X])\cdot \gamma}{\Var(X_N)}
    \\& = \gamma \mathbb{E}[X]\mathbb{E}[Y] +\Cov(X_N,Y_N).
\end{align*}
On the other hand, using that $Y^1= \frac{\mathbb{E}[XY]}{\mathbb{E}[X]}$ (since $\mathbb{E}[XY] = \mathbb{E}[Y|X=1]\mathbb{E}[X]$) we can write that:
\begin{align*}
    \gamma \mathbb{E}[X] \cdot Y^1 = \gamma \mathbb{E}[XY].
\end{align*}
Putting together what we have so far, we have:
\begin{align*}
    \frac{\Var(X_n)}{1-\mathbb{E}[X]} \cdot \left(Y_{ER}^1 - Y^1\right) &= \gamma\mathbb{E}[X]\mathbb{E}[Y]+\Cov(X_N,Y_N)-\gamma\mathbb{E}[XY]
    \\&= \Cov(X_N,Y_N)-\gamma \Cov(X,Y).
\end{align*}
But using the definition of $\gamma$, the fact that $D=\frac{\Cov(X,Y)}{\Var{X}}$, and the fact that $D= D_{ER}-\frac{\delta_W}{\Var(X_N)}$ by Proposition \ref{prop:bias_er}, we have that:
\begin{align*}
    \Cov(X_N,Y_N)-\gamma\Cov(X,Y) &=  \Cov(X_N,Y_N) - \frac{\Cov(X,Y)}{\Var(X)}\cdot \Var(X_N)
    \\&= \Cov(X_N,Y_N) - \Var(X_N)\cdot D \\&= \Cov(X_N,Y_N)- \Var(X_N)\left[D_{ER} -\frac{\delta_W}{\Var(X_N)}\right]
    \\&= \delta_W 
\end{align*}
where the final equality follows by the fact that $D_{ER} = \Cov(X_N,Y_N)/\Var(X_N)$. This proves the claim. 

The second part follows similarly. Write $\Var(X_N)/\mathbb{E}[X] = \gamma (1-\mathbb{E}[X])$. Then:
\begin{align*}
    Y^0\cdot \gamma \cdot(1-\mathbb{E}[X]) = \frac{\mathbb{E}[(1-X)Y]}{1-\mathbb{E}[X]}\cdot \gamma \cdot (1-\mathbb{E}[X])=\gamma \mathbb{E}[Y]-\gamma\mathbb{E}[XY]
\end{align*}
where we have used that $\mathbb{E}[(1-X)Y]=\mathbb{E}[Y|X=0]\cdot(1-\mathbb{E}[X])$. 
And:
\begin{align*}
    \gamma (1-\mathbb{E}[X]) \cdot Y_{ER}^{0} &= \gamma(1-\mathbb{E}[X])\mathbb{E}[Y_N] - \gamma (1-\mathbb{E}[X])\frac{\Cov(X_N,Y_N)}{\Var(X_N)}\cdot \mathbb{E}[X]\\
    &= \gamma \mathbb{E}[Y]-\gamma\mathbb{E}[X]\mathbb{E}[Y] - \Cov(X_N,Y_N).
\end{align*}
Putting these together:
\begin{align*}
    \frac{\Var(X_N)}{\mathbb{E}[X]}\cdot \left(Y_{ER}^{0}-Y^0\right) &= \gamma\mathbb{E}[Y]-\gamma\mathbb{E}[X]\mathbb{E}[Y]-\Cov(X_N,Y_N)-(\gamma\mathbb{E}[Y]-\gamma\mathbb{E}[XY])
    \\&= \gamma \Cov(X,Y)-\Cov(X_N,Y_N)
    \\&=-\delta_W,
\end{align*}
where the last equation follows similarly by Proposition \ref{prop:bias_er} and the analogous facts as those mentioned at this point in the proof for $Y_{ER}^1-Y^1$ above. This completes the proof.  
\end{proof}

\begin{proof}[Proof of \textbf{Proposition \ref{prop:bias_nm_levels}}]
Notice that:
\begin{align*}
    \mathbb{E}[XY] &= \Pr[X=1 \cap Y=1] = \mathbb{E}[Y|X=1]\Pr[X=1]
\end{align*}
We can thus write that:
\begin{align*}
    Y^1 = \frac{\mathbb{E}[XY]}{\mathbb{E}[X]}
\end{align*}
Then
\begin{align*}
    Y_{NM}^1 -Y^1= \frac{\mathbb{E}[X_N Y_N]}{\mathbb{E}[X_N]} - \frac{\mathbb{E}[XY]}{\mathbb{E}[X]} = \frac{\mathbb{E}[X_N Y_N]-\mathbb{E}[XY]}{\mathbb{E}[X]}
\end{align*}
since $\mathbb{E}[X_N]=\mathbb{E}[X]$. Now we note that $\mathbb{E}[X_N Y_N] = \Cov(X_N, Y_N)+\mathbb{E}[X_N]\mathbb{E}[Y_N] = \Cov(X_N, Y_N) +\mathbb{E}[X]\mathbb{E}[Y]$, and substitute this in to obtain:
\begin{align*}
    \frac{\mathbb{E}[X_N Y_N]-\mathbb{E}[XY]}{\mathbb{E}[X]} = \frac{\Cov(X_N,Y_N) +\mathbb{E}[X]\mathbb{E}[Y] - \mathbb{E}[XY]}{\mathbb{E}[X]}=\frac{\Cov(X_N, Y_N)-\Cov(X,Y)}{\mathbb{E}[X]}
\end{align*}
Now applying the law of total covariance, we can write that:
\begin{align*}
\Cov(X,Y)=\mathbb{E}[\Cov(X,Y|N)]+\Cov(\mathbb{E}[X|N],\mathbb{E}[Y|N])=\delta_B + \Cov(X_N,Y_N)
\end{align*}
and plugging this back in gives:
\begin{align*}
    \frac{\Cov(X_N, Y_N)-\Cov(X,Y)}{\mathbb{E}[X]} = \frac{\Cov(X_N, Y_N)-\delta_B -\Cov(X_N,Y_N)}{\mathbb{E}[X]} =-\frac{\delta_B}{\mathbb{E}[X]}
\end{align*}
as desired. To see second part, note that:
\begin{align*}
    \mathbb{E}[(1-X)Y]= (1-\mathbb{E}[X])\mathbb{E}[Y|X=0] \implies Y^0 = \frac{\mathbb{E}[(1-X)Y]}{1-\mathbb{E}[X]}
\end{align*}
so subtracting $Y^0$ from $Y_{NM}^0$ yields
\begin{align*}
    \frac{\mathbb{E}[(1-X_N)Y_N]}{\mathbb{E}[1-X_N]}-\frac{\mathbb{E}[(1-X)Y]}{\mathbb{E}[1-X]} &= \frac{\mathbb{E}[Y_N]-\mathbb{E}[X_N Y_N] - \mathbb{E}[Y] + \mathbb{E}[XY]}{1-\mathbb{E}[X]}\\&= \frac{\mathbb{E}[XY]-\mathbb{E}[X_N Y_N]}{1-\mathbb{E}[X]}.
\end{align*}
The numerator is now the negative of the term arrived at before invoking the definition of covariance in the proof for $Y^1$, so can follow similarly, and we arrive at:
\begin{align*}
    Y_{NM}^{0}-Y^{0}= \frac{\delta_B}{1-\mathbb{E}[X]}
\end{align*}
as desired. 
\end{proof}

\begin{proof}[Proof of \textbf{Theorem \ref{thm:identification_Y1}}]
The ideas in this proof are very similar to that of Theorem \ref{thm:identification_D}, so we will sketch out the key points without belaboring the details.  
    If $\delta_W>0$, then:\begin{align*} Y^{1} = Y_{ER}^1-\frac{\delta_W(1-\E[X])}{\Var(X_N)} \leq Y_{ER}^{1}
    \end{align*}
    If $\delta_B>0$, then 
    \begin{align*}
        Y^1 = Y_{NM}^1 + \frac{\delta_B}{\mathbb{E}[X]}\geq Y_{NM}^1
    \end{align*}
    Thus under CR with $\delta_W,\delta_B>0$, we have that
    \begin{align*}
    Y_{NM}^1\leq Y_{NM}^1+\frac{\delta_B}{1-\mathbb{E}[X]} \leq Y^1 \leq Y_{ER}^1-\frac{\delta_W(1-\mathbb{E}[X])}{\Var(X_N)}\leq Y_{ER}^1
    \end{align*}
    Note that we must also have always $Y^1 \leq Y_{MOB}^{1+}$ and $Y^1\geq Y_{MOB}^{1-}$ by the definition of the method of bounds; the neighborhood model is guaranteed to have $Y^1$ within the method of bounds estimates,, so we can omit those when the boundary is  $Y_{NM}^1$, but we do need to compare it to $Y_{ER}^1$. 

    Similarly, if $\delta_W<0$, then $Y^1 \geq Y_{ER}^1$, if $\delta_B<0$ then $Y_1\leq Y_{NM}^1$, and if both are true -- i.e. contextual reinforcement in the negative direction -- then $Y_{ER}^1\leq Y^1\leq Y_{NM}^1$. 

    For the case in which $\delta_B \leq 0, \delta_W \geq 0$, note that in this case we will have $Y^1\leq Y_{NM}^1$, $Y^1 \leq Y_{ER}^1$, and as always $Y^1 \leq Y_{MOB}^{1+}$, meaning it must be less than the minimum of the three; similarly with the maximum of $Y_{ER}^1$,$Y_{NM}^1$, and $Y_{MOB}^{1-}$ in the opposite case. 

    In cases where $\delta_W$ or $\delta_B$ are known to be exactly equal, then the estimator coincides with the truth. Finally both effects are $0$, then by Theorem \ref{thm:identification_D} we must have that $D=0$, which implies $Y^1=Y^0$. 

    And finally with regards to sharpness, we can use a similar argument as before to obtain any $Y^{1'}$ within the boundaries using a mixture; hence, the bounds are sharp. 
    
\end{proof}

\begin{proof}[Proof of \textbf{Theorem \ref{thm:identification_Y0}}]
    The proof is similar to that of Theorem \ref{thm:identification_Y1}, but with key signs reversed due to the sign difference in the biases laid out in Propositions \ref{prop:bias_er_levels} and \ref{prop:bias_nm_levels}. In particular, note that $\delta_B \geq0 \implies Y^0=Y_{NM}^0-\frac{\delta_B}{1-\mathbb{E}[X]}\leq Y_{NM}^0$, while $\delta_W\geq0 \implies Y^{0}=Y_{ER}^0 +\frac{\delta_W\mathbb{E}[X]}{\Var(X_N)}\geq Y_{ER}^{0}$. The rest of the proof follows similarly. 
\end{proof}

\begin{proof}[Proof of \textbf{Proposition \ref{prop:cga_long_bounds}}]
(i) Suppose that $\delta_{B,n}\geq0$. Then $Y_n^1\geq Y_n^0$, immediately implying that $D_n\geq 0$. Of course, $D_n$ still must be below $D_{MOB,n}^+$, and combining these two facts yields the stated interval. Now, to see that $Y_n^1\geq Y_n$ and $Y_n^0 \leq Y_n$, note that $Y_n = X_n Y_n^1 + (1-X_n) Y_n^0$ is a weighted average of $Y_n^1$ and $Y_n^0$. Thus, $Y_n^1 \geq Y_n^0$ implies that $Y_n^0 \leq Y_n \leq Y_n^1$. These facts give the lower and upper bound on $Y_n^1$ and $Y_n^0$ respectively, and combining with the method of bounds yields the stated interval.

(ii) If $\delta_{B,n} \leq 0$, then similarly $Y_n^1 \leq Y_n^0$, implying $D_n \leq 0$ while also $D_n \geq D_{MOB,n}^{-}$. Again, $Y_n$ being a weighted average of $Y_n^1$ and $Y_n^0$ gives the remainder of the statement.

(iii) For sharpness for (i), we need to show that each of the stated values do not contradict the assumptions. As usual, since everything is contained within the method of bounds, the only possibility that we need to consider is that some point in the stated intervals could \emph{require} $\delta_{B,n}<0$ and thus not be consistent with the assumed hypothesis. But since $\delta_{B,n} = D_n \Var[X|N=n]$, showing that a given point requires $\delta_{B,n}<0$ implies this point requires $D_n <0$. Thus, if we can show all points are consistent with  $D_n \geq 0$, they are consistent with $\delta_{B,n}\geq 0$ and are feasible. 

Start with $D_n$, which is always within the method of bounds. As just mentioned, $D_n =0$ if and only if $\delta_{B,n}=0$, so $D_n=0$ is feasible. For the other end of the interval, as long as $D_{MOB}^{+}\geq 0$, it does not rule out $D_{n} \geq 0$ which again is true if and only $\delta_{B,n}\geq 0$. But since $D_n=0$ is always feasible, the method of bounds $D_{MOB,n}^{+}$ \emph{must} at least include $0$, so $D_{MOB,n}^{+}\geq0$. Thus, $D_{MOB,n}^{+}$ is feasible as well. Applying Corollary \ref{lem:feasb_intermediate_disp} then shows that any $D$ in between $0$ and $D_{MOB,n}^{+}$ is also feasible, and the claim that the interval is sharp is proved. 

Now consider $Y_n^1$. $Y_n^1 = Y_n$ is feasible because setting $Y_n^1=Y_n=Y_n^0$ creates feasible disparity $D_n=0$. We can also see that $Y_n^1=Y_{MOB,n}^{1+}$ (paired with the complement  $Y_n^0 = Y_{MOB,0}^-$) is also feasible, because it corresponds to $D_n = D_{MOB,n}^{+} \geq 0$.

Since both $Y_n^1=Y_n$ and $Y_n^1=Y_{MOB,n}^{1+}$ are feasible; Lemma \ref{lem:feasb_intermediates} then gives that all points in the range are similarly feasible, and Corollary \ref{lem:feasb_intermediate_disp} shows that the disparity must also be nonnegative. The reasoning for the range containing $Y_n^{0}$ is symmetric. Finally, the argument for sharpness for (ii) is similar, mutatis mutandis. 

\end{proof}

\begin{proof}[Proof of \textbf{Proposition \ref{prop:cna_long_bounds}}]
    \ 
     \\(i) If $\delta_{W,n}\geq0$, then $D_n = \mu'(X_n)-\delta_{W,n} \leq \mu'(X_n)$. Since the method of bounds still holds, we have that $D_n \leq \min\{\mu'(X_n), D_{MOB,n}^+\}$; we also still have $D_n \geq D_{MOB,n}^{-}$, giving the claimed interval. 
     
     Turning to $Y_n^1$: note that rearranging  $D_n = Y_n^1-Y_n^0$ and substituting $Y_n=X_n Y_n^1 + (1-X_n) Y_n^0$ in shows that $Y_n^1 = Y_n + (1-X_n) \cdot D_n$. But since $D_n \leq \mu'(X_n) $ given that $\delta_{W,n}\geq 0$, we have that $Y_n^1 \leq Y_n + (1-X_n)\cdot \mu'(X_n)$. Combining with the MOB, we get $Y_n^1 \in [Y_{MOB,n}^{1-},\min\left\{Y_{MOB,n}^{1+},Y_n +(1-X_n)\cdot\mu'(X_n)\right\}$. Rearranging $D= Y_n^1 - Y_n^0$ the other way and substituting into $Y_n^0 = Y_n - X_n D_n$ gives that $Y_n^0 = Y_n-X_n\cdot D_n \geq Y_n - X_n \cdot \mu'(X_n)$, with the inequality again following since $\mu'(X_n)\geq D_n$, and combining this with the MOB again shows that $Y_n^0 \in [\max\left\{Y_{MOB,n}^{0-},Y_n-X_n\cdot \mu'(X_n)\right\}, Y_{MOB,n}^{0+}]$. 
     
     (ii) If $\delta_{W,n}\leq 0$, reversing the inequalities and applying the same substitutions gives the stated bounds.
     
     (iii) We will show sharpness for (i); (ii) again follows by using the same substitutions and reversing the inequalities. Like in Proposition \ref{prop:cga_long_bounds}, we need to show that for any given point in the range, there exists a feasible distribution with $\delta_{W,n} \geq 0$; note that this is true if and only if $D_n \leq \mu'(X_n)$; so in particular, to show that a point is feasible, we must show that there exists a distribution consistent with the data such that $D_n \leq \mu'(X_n)$.

     Note first off that there are two cases  depending on how $\mu'(X_n)$ relates to $D_n$. If $D_{MOB}^{+} \leq \mu'(X_n)$, then the claimed interval reduces to the method of bounds interval;  by construction all points in this range are consistent with $Y_n$ and $\overline{Y},\underline{Y}$; Moreover, by Lemma \ref{lem:feasb_intermediates}, everything in that range has disparity $D_n \leq D_{MOB}^{+}\leq \mu'(X_n)$, and so satisfies $\delta_{W,n}\geq0$; thus, that entire range is feasible and the bound is sharp in this case. 

     So now suppose that $\mu'(X_n) < D_{MOB}^{+}$. The range reduces to a subset of the method of bounds now, so again every point in the range is feasible in terms of $Y_n$ and $\overline{Y},\underline{Y}$. To see that every point between $[D_{MOB}^-,\mu'(X_n)]$ is feasible in terms of $D_n \leq \mu'(X_n)$, again apply Lemma \ref{lem:feasb_intermediates} with the choice of extreme distributions $(Y_{MOB}^-,Y_{MOB}^{0+})$ and $(Y^{1\mu},Y^{0\mu})$, where $Y^{1\mu}$ and $Y^{0\mu}$ are marginals give rise to $D_{\mu}= \mu'(X_n)$ — such a pair must exist by Corollary \ref{lem:feasb_intermediate_disp} since $\mu'(X_n)$ is within the method of bounds, and in fact can explicitly be constructed as $Y^{1\mu}=Y_n + (1-X_n)\mu'(X_n)$, $Y^{0\mu} = Y_n-X_n \mu'(X_n)$; this distribution has disparity $\mu'(X_n)$, and it can be shown that these must both be in $[\underline{Y},\overline{Y}]$.) Thus, the bound is sharp for $Y_n^1$. A similar argument follows for $Y_n^0$.

 \end{proof}   

\begin{proof}[Proof of \textbf{Proposition \ref{prop:cr_local}}]
    This follows by applying Propositions \ref{prop:cga_long_bounds} and \ref{prop:cna_long_bounds} simultaneously. We shall prove (i) for the case that $\delta_{B,n}\geq 0$ and $\delta_{W,n} \geq 0$; the case that $\delta_{B,n} \leq 0 $ and $\delta_{W,n} \leq 0$ is similar. We will omit (ii), because it is similar to (i), mutatis mutandis, and then prove (iii), i.e. sharpness, again for the case that $\delta_{B,n}\geq 0$ and $\delta_{W,n} \geq 0$. 
    
    (i) Suppose that $\delta_{B,n}\geq 0$ and $\delta_{W,n} \geq 0$. Notice that $\delta_{B,n}\geq0 \implies D_n\geq 0$, and this plus $\delta_{W,n}\geq 0 $ means that $\mu'(X_n) = D_n + \delta_{W,n}\geq 0$. 
    Applying the assumptions (i) from both Propositions \ref{prop:cga_long_bounds} and \ref{prop:cna_long_bounds} gives that $D_n \in [D_{MOB,n}^-, \min\left\{\mu'(X_n),D_{MOB,n}^+\right\}]$ and $D_n \in [0, D_{MOB,n}^+]$; that $D_n \in [0,\min\left\{\mu'(X_n),D_{MOB,n}^+\right\}]$ then follows since both must hold. (Notice that $\mu'(X_n)\geq 0$ implies the claimed interval is non-empty.) $Y_n^1$ and $Y_n^0$ follow similarly.  
    
    (ii) As noted, (ii) follows similarly to (i).

    (iii) For sharpness, we must show all points in the range do not contradict both 1) $\delta_{B,n}\geq0$ and 2) $\delta_{W,n}\geq 0$. Recall that $\delta_{B,n} \geq 0 \iff D_n\geq0$, and $\mu'(X_n) = D_n +\delta_{W,n} \implies \delta_{W,n} \geq 0 \iff \mu'(X_n) \geq D_n$. So all we need to do to show a given value is feasible is that is consistent with some distribution where $D_n \geq0$ and $D_n \leq \mu'(X_n)$. Note that $Y_n^1=Y_n^0=Y_n$ gives $D_n=0$, and $(Y^{1\mu}, Y^{0\mu})$ as described above gives $D_n=\mu'(X_n)$.  But now by invoking Lemma \ref{lem:feasb_intermediates} and Corollary \ref{lem:feasb_intermediate_disp} in the same way as we did in the proofs of Propositions \ref{prop:cga_long_bounds} and \ref{prop:cna_long_bounds}, we can see that every $D_n$ in $[0, \mu'(X_n)]$ corresponds to some feasible distribution, and the bounds are thus sharp.

\end{proof}

\subsection*{Relaxing Assumption 1}
Now we turn to relaxing Assumption \ref{assmp:exclusion}. In Assumption \ref{assmp:exclusion}, we are assuming that any two neighborhoods that have the same group prevalence will also share the same expected outcome by group; relaxing this means allowing that neighborhoods with the same group prevalence have different expected outcomes. In this case, the analogues to Propositions \ref{prop:cga_long_bounds} -\ref{prop:cr_local} apply in an average sense; i.e. we can obtain upper and lower bounds on the average disparity and average values of $Y_n^j$ over the set of neighborhoods with the same group prevalence. 

For each of our random variables, we use the notation $\tilde{\cdot}_n$ to indicate that we are averaging over the set of neighborhoods $n' \in \mathcal{N}$ such that $X_{n'}=X_n$. For example:
\begin{align*}
    \tilde{Y}_n^1 :&= \E[Y_n^1|n\in\mathcal{N}^{-1}(X_n)]\\&
    =\sum_{n \in \mathcal{N}^{-1}(x_n)} w_n Y_n^1
\end{align*}
where $\mathcal{N}^{-1}(x_n)$ is the set of $n\in\mathcal{N}$ such that $X_n=x_n$ and $w_n = p_n/\sum_{n\in\mathcal{N}^{-1}(x_n)} p_n$. We define $\tilde{Y}_n^0$ similarly, and  $\tilde{D}_n$ as $\tilde{Y}_n^1-\tilde{Y}_n^0$. We can also define the method of bounds similarly — $\tilde{Y}_{MOB,n}^{1+} = \sum_{n \in \mathcal{N}^{1+}(x_n)} w_n 
 Y_{MOB,n}^{1+}$ and so on, and immediately have that:
\begin{align*}
    \tilde{Y}_n^1 \in [\tilde{Y}_{MOB,n}^{1-},\tilde{Y}_{MOB,n}^{1+}], \\ \tilde{Y}_n^0 \in [\tilde{Y}_{MOB,n}^{0-},\tilde{Y}_{MOB,n}^{0+}],
\end{align*}
because 
\begin{align*}
    Y_n^1 \leq Y_{MOB,n}^{1+} \ \forall n \in \mathcal{N}^{-1}(x_n) \implies \sum_{n\in\mathcal{N}^{-1}(x_n)} w_n Y_{n}^1 \leq \sum_{n \in \mathcal{N}^{-1}(x_n) }w_n Y_{MOB,n}^{1+} =\tilde{Y}_{MOB,n}^{1+}
\end{align*}
and so on; similarly, we can write that: 
\begin{align*}
    \tilde{D}_n \in \  &[\tilde{D}_{MOB,n}^{-},\tilde{D}_{MOB,n}^{+}] \\=& [\tilde{Y}_n^{1-}-\tilde{Y}_n^{0+},\tilde{Y}_{n}^{1+}-\tilde{Y}_n^{0-}].
\end{align*}

These define the method of bounds. We now turn to the equivalents of Propositions \ref{prop:cga_long_bounds}-\ref{prop:cr_local}. As before, define $\mu(x_n):= \E[Y|X_n=x_n]$. Define also
\begin{align*}
    \tilde{\delta}_{B,n} := \Var[X|N=n]\left(\tilde{Y}_n^1-\tilde{Y}_n^0\right)
\end{align*}
and
\begin{align*}
    \tilde{\delta}_{W,n} := \sum_{n \in \mathcal{N}^{-1}(X_n)} w_n\tilde{\delta}_{W,n}.
\end{align*}

Note first:
\begin{lemma}
Define $\tilde{\delta}_{W,n}$, $\tilde{D}_n$ as above. Then:
\begin{align*}
    \tilde{D}_n = \mu'(X_n)-\tilde{\delta}_{W,n}
\end{align*}
    
\end{lemma}

\begin{proof}

Note that:
\begin{align*}
    \mu(x_n) = \sum_{n\in\mathcal{N}^{-1}(x_n)} w_n Y_n;
\end{align*}
 this can be re-written as:
\begin{align*}
    \mu(x_n) &= \sum_{n\in\mathcal{N}^{-1}(x_n)} w_n \cdot(X_n Y_n^1 +(1-X_n)Y_n^0) \\ &= X_n \cdot\sum_{n\in\mathcal{N}^{-1}(x_n)} w_n Y_n^1 + (1-X_n)\cdot \sum_{n\in\mathcal{N}^{-1}(x_n)} w_n Y_n^0
\end{align*}
since $X_n$ is constant for all $n\in\mathcal{N}^{-1}(x_n)$. Then taking the derivative gives:
\begin{align*}
    \mu'(X_n) &= X_n     \cdot \sum_{n\in\mathcal{N}^{-1}(x_n)} w_n Y_n^{1'} + (1-X_n)\cdot\sum_{n\in\mathcal{N}^{-1}(x_n)} w_n Y_n^{0'} + \sum_{n\in\mathcal{N}^{-1}(x_n)} w_n (Y_n^1 - Y_n^0)
    \\&= \sum_{n\in\mathcal{N}^{-1}(x_n)}w_n\left(X_n\mu_1'(x_n)+(1-X_n)\mu_0'(x_n)\right) + \sum_{n \in \mathcal{N}^{-1}(x_n)}w_n(Y_n^1-Y_n^0)
    \\&= \sum_{n \in \mathcal{N}^{-1}(x_n)}w_n \delta_{W,n}+\sum_{n \in \mathcal{N}^{-1}(x_n)} w_n (Y_n^1-Y_n^0) \\&=\tilde{\delta}_{W,n} + \tilde{D}_n.
\end{align*}
\end{proof}

We can now state the generalized versions of Propositions \ref{prop:cga_long_bounds}-\ref{prop:cr_local}:
\begin{prop}[General version of Proposition \ref{prop:cga_long_bounds}]\label{prop:cga_long_bounds_gen}
(i)
If $\tilde{\delta}_{B,n}\geq 0$ for some set of neighborhoods $\mathcal{N}^{-1}(x_n)$, then: 
\begin{align*}
    \tilde{D}_n & \in [0, \tilde{D}_{MOB,n}^{+}]\\
    \tilde{Y}_n^1 & \in [\tilde{Y}_n, \tilde{Y}_{MOB,n}^{1+}]\\
    \tilde{Y}_n^0 &\in [\tilde{Y}_{MOB,n}^{0-},\tilde{Y}_{n}]
    \end{align*}
(ii) If $\tilde{\delta}_{B,n} \leq 0$ for some set of neighborhoods $\mathcal{N}^{-1}(x_n)$, then:
\begin{align*}
    \tilde{D}_n &\in [\tilde{D}_{MOB,n}^{-},0],\\
    \tilde{Y}_n^1 & \in[\tilde{Y}_{MOB,n}^{1-},\tilde{Y}_n] \\
    \text{ and } \tilde{Y}_n^0 &\in [\tilde{Y}_n, \tilde{Y}_{MOB,n}^{0+}].
\end{align*}
(iii) The bounds in (i) and (ii) are sharp in the absence of additional information.
\end{prop}
\begin{prop}[General version of Proposition \ref{prop:cna_long_bounds}]
\label{prop:cna_long_bounds_gen}

(i)Suppose that $\tilde{\delta}_{W,n}\geq 0$ for some neighborhoods $\mathcal{N}^{-1}(x_n)$. Then:
\begin{align*}
\tilde{Y}_n &\in \left[\tilde{D}_{MOB,n}^{-},\min\left\{\mu'(X_n),\tilde{D}_{MOB,n}^{+} \right\} \right],\\
Y_n^1 & \in
\left[  \tilde{Y}_{MOB,n}^{1-}, \min \left\{ \tilde{Y}_{MOB,n}^{1+}, \tilde{Y}_n+(1-X_n)\mu'(X_n) \right\}
\right],\\
\text{ and } \tilde{Y}_n^0 &\in \left[\max\left\{\tilde{Y}_{MOB,n}^{0-},\tilde{Y}_n - X_n \cdot \mu'(X_n)\right\}\right]
\end{align*}
(ii) Suppose that $\tilde{\delta}_{W,n}\leq 0 $ for some set of neighborhoods $\mathcal{N}^{-1}(x_n)$. Then:
\begin{align*}
    \tilde{Y}_n &\in \left[ \max \left\{\mu'(x_n),\tilde{D}_{MOB,n}^{-} \right\},\tilde{D}_{MOB,n}^{+}\right]
    \\
    \tilde{Y}_n^1 &\in \left[ \max \left\{ \tilde{Y}_{MOB,n}^{1-},\tilde{Y}_n+(1-X_n)\cdot\mu'(X_n)\right\},\tilde{Y}_{MOB,n}^{1+} \right]\\
     Y_n^0 &\in \left[ \tilde{Y}_{MOB,n}^{0-}, \min\left\{\tilde{Y}_{MOB,n}^{0+},\tilde{Y}_n-X_n\cdot\mu'(X_n)\right\} \right]
\end{align*}

(iii) The bounds in (i) and (ii) are sharp in the absence of further information.
\end{prop}

\begin{prop}[General Version of Proposition \ref{prop:cr_local}]\label{prop:cr_local_gen}
    Suppose that $\tilde{\delta}_{B,n}\cdot \tilde{\delta}_{W,n}\geq 0$. Then either:

(i) $\mu'(X_n) \geq 0$, and:
\begin{align*}
    \tilde{D}_n &\in [0, \min\{\mu'(X_n),\tilde{D}_{MOB,n}^{+}]\\
    Y_n^1 & \in \left[\tilde{Y}_n,\min\{\tilde{Y}_n+(1-X_n)\cdot \mu'(X_n),\tilde{Y}_{MOB,n}^{1+}\right]\\
    \tilde{Y}_n^0 &\in \left[\tilde{Y}_n,\min\{\tilde{Y}_n -X_n\mu'(X_n),\tilde{Y}_{MOB,n}^{0+}\right];
\end{align*}
or
(ii) $\mu'(X_n) \leq 0$, and 
\begin{align*}
    \tilde{D}_n &\in \left[\max\{\mu'(X_n), \tilde{D}_{MOB,n}^{-},0\right],\\
    \tilde{Y}_n^1 &\in [\max\{\tilde{Y}_{MOB,n}^{1-},\tilde{Y}_n+(1-X_n)\mu'(X_n)\},\tilde{Y}_n]\\
    \text{ and } Y_n^0 & \in \left[\tilde{Y}_n, \min\{\tilde{Y}_n-X_n\cdot \mu'(X_n),\tilde{Y}_{MOB,n}^{0+}\}\right]
\end{align*}

(iii) The bounds in (i) and (ii) are sharp in the absence of further information. 
\end{prop}

Proofs of these propositions follow similarly to those of the analogous special cases having replaced each quantity with its averaged counterpart. 

\newpage
\section{Supplemental Figures/Tables}
\label{sec:supplement}

\renewcommand{\thefigure}{S\arabic{figure}}
\setcounter{figure}{0}

\renewcommand{\thetable}{S\arabic{table}}
\setcounter{table}{0}

\begin{figure}[ht]
    \begin{centering}
    \caption{Identification of Partisan Vaccination Rate}
    \label{fig:bounds_by_group}
    \includegraphics[width=0.9\textwidth]{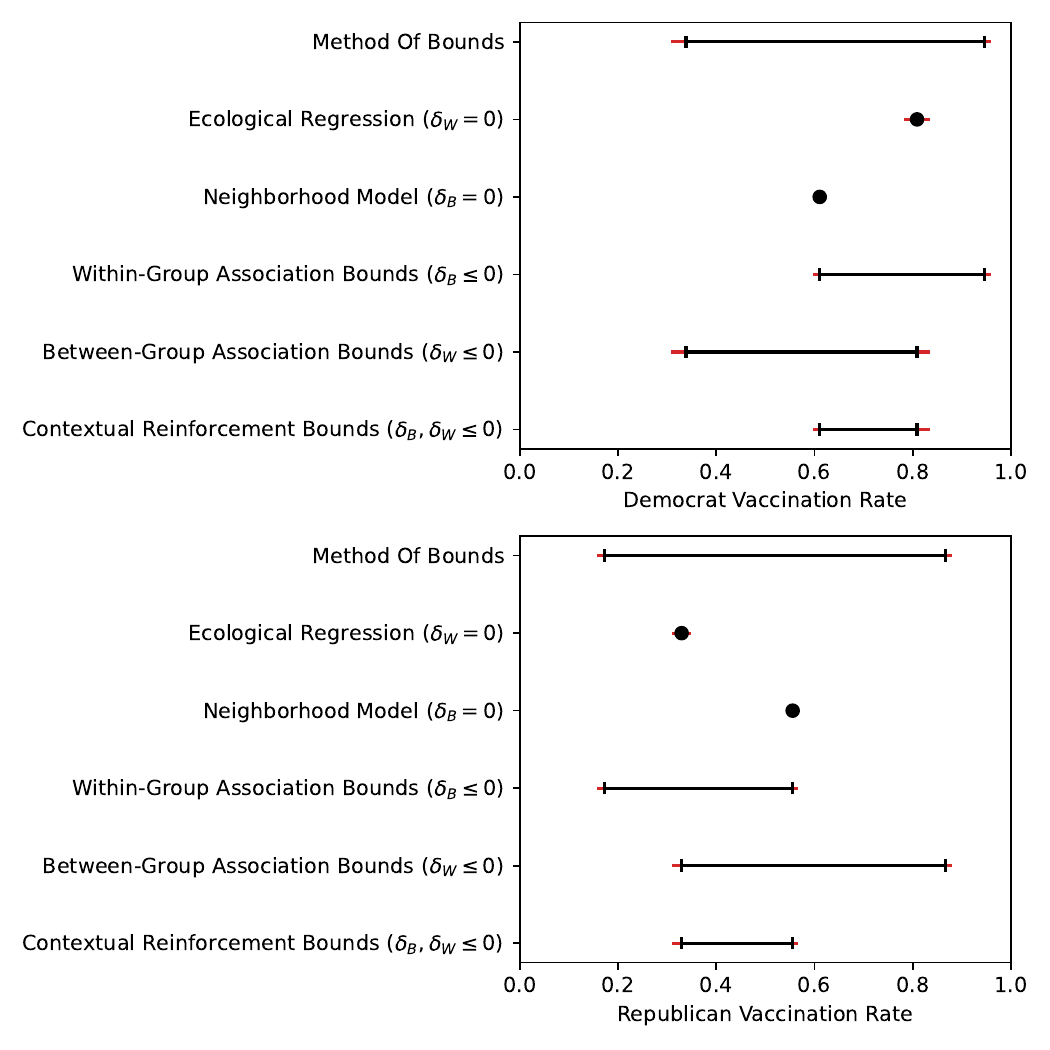}
    \end{centering}
    \vspace{-5pt} \\
    {\footnotesize{Notes: Mean vaccination rate for Democrats and Republicans (components of the partisan vaccination gap). Red bars are 95\% confidence intervals following \citep{imbens2004confidence}; the confidence intervals are based on standard errors from a county-level bootstrap with 1000 bootstrap replicates.}}
\end{figure}

\begin{figure}[ht]
    \begin{centering}
    \caption{Median Income and Proportion Rural Population by Republican Vote Share}
    \label{fig:rural_income_republican}
    \includegraphics[width=0.8\textwidth]{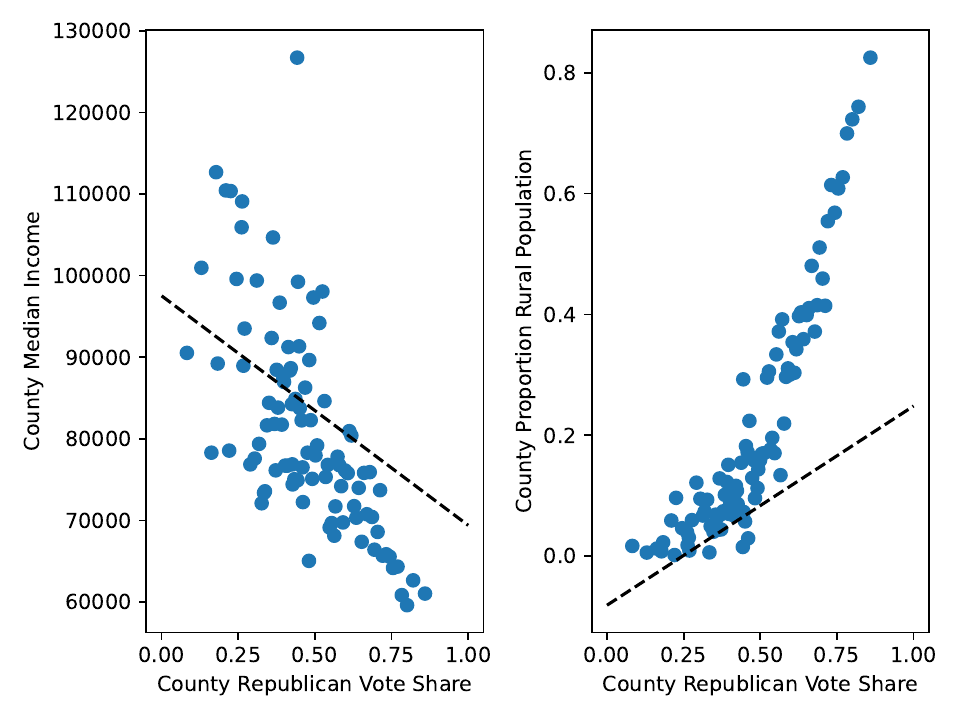}
    \end{centering}
    \vspace{-5pt} \\
    {\footnotesize{Notes: The figure reports county-level median income and proportion of the population living in rural areas. Counties are grouped into 100 equal-population bins (blue). Population-weighted linear regressions are also shown (black dotted lines). Household income data and proportion rural population were sourced from the 2023 5-year American Community Survey and the 2020 Decennial Census, respectively. Republican vote share was sourced from the MIT Election Data and Science Lab.}}
\end{figure}

\begin{figure}[ht]     
\caption{Local Linear Approximation to County Vaccine Data}   
\label{fig:local_linear}
\begin{centering} 
\includegraphics[width=0.8\textwidth]{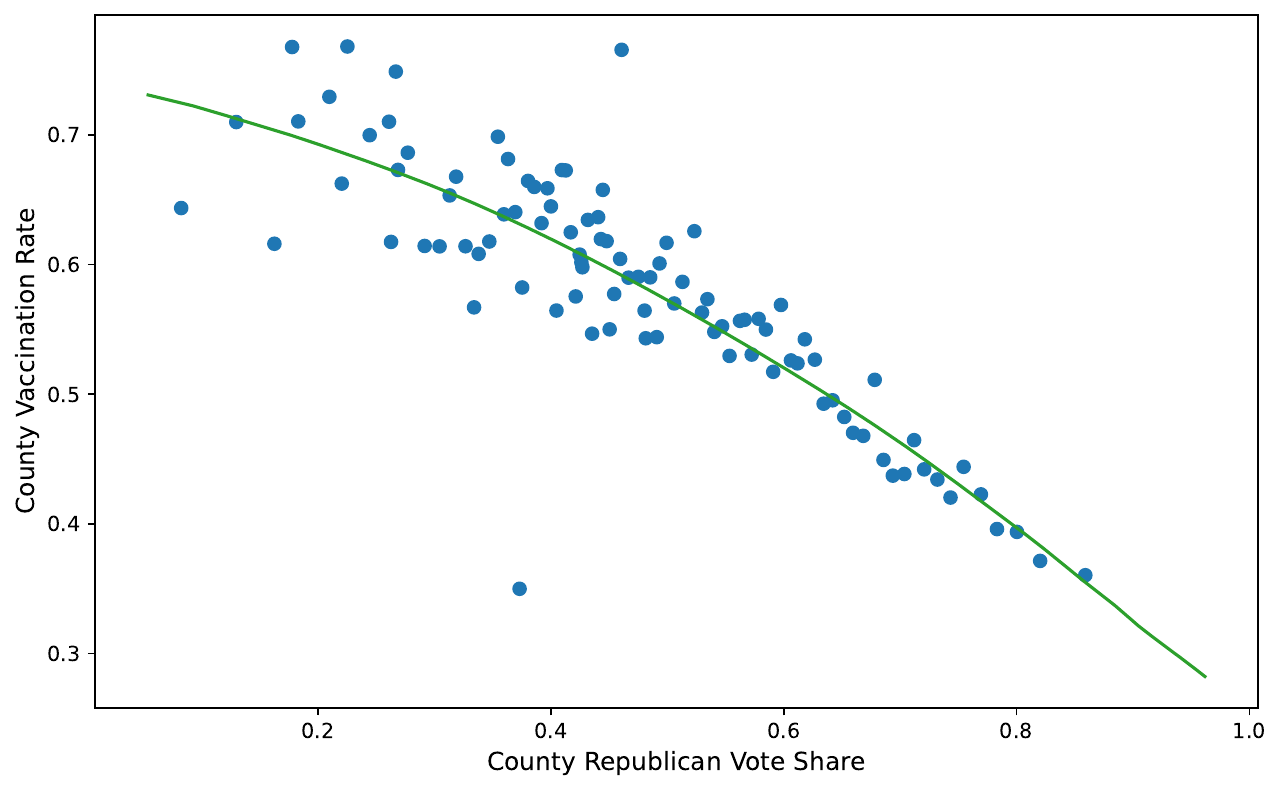}
    \vspace{-5pt} \\
    \end{centering}
{\footnotesize{Notes: This figure reports county-level COVID-19 vaccination rates by the share of voters in the county who voted for the Republican candidate in the 2020 presidential election. Counties are grouped into 100 equal-population bins (blue). The local linear approximation for county-level data is shown in green. This approximation uses an Epanechnikov kernel and 10-fold cross-validation to choose the bandwidth parameter.}}
\end{figure}

\begin{longtable}{|c|c|c|c|}
\hline
\textbf{Tier} & \textbf{Uniquely Matched} & \textbf{Unmatched} & \textbf{Proportion Matched} \\
\hline
Tier 1 & 1,028,172  & 764,810 & .569 \\
Tier 2 & 61,906 & 712,439 & .603 \\
Tier 3 & 129,298 & 582,697 & .675\\
Tier 4 & 37,409 & 545,694 & .695 \\
Tier 5 & 35,231 & 510,427 & .715 \\
Tier 6 & 89,605 & 420,430 & .764 \\
\hline
\caption{Summary of matched patient counts for the auxiliary individual-level dataset. Tier 1 contains the highest confidence matches and Tier 6 contains the lowest confidence matches (for details, see Appendix \ref{sec:app_application}). For individuals with multiple matches, one match is chosen at random.}
\label{tab:patient_tiers_randomtie}
\end{longtable}

\renewcommand{\arraystretch}{2.5}
\begin{longtable}{|c|c|c|c|c|c|c|}
    \hline
    & \textbf{Tier 1} & \textbf{Tier 1-2} & \textbf{Tier 1-3} & \textbf{Tier 1-4} & \textbf{Tier 1-5} & \textbf{Tier 1-6} \\
    \hline
    \shortstack{$E[Cov(Y, X_N|X)]$} & 
    \shortstack{-0.00463 \\ (8e-05)} & 
    \shortstack{-0.00452 \\ (8e-05)} & 
    \shortstack{-0.00451 \\ (7e-05)} & 
    \shortstack{-0.00447  \\ (7e-05)}&
    \shortstack{-0.00455  \\ (7e-05)}& 
    \shortstack{-0.00466  \\ (7e-05)} \\
    \hline
    \shortstack{$E[Cov(Y, X|N)]$} & 
    \shortstack{-0.00691 \\ (0.00022)} & 
    \shortstack{-0.00703 \\ (0.00021)} & 
    \shortstack{-0.0066 \\ (0.0002)} & 
    \shortstack{-0.00652  \\ (0.0002)}&
    \shortstack{-0.00623  \\ (0.0002)}& 
    \shortstack{-0.00592  \\ (0.0002)} \\
    \hline
\caption{Conditional covariance (standard error) estimated from the auxiliary individual-level dataset. $Y$ indicates whether an individual is vaccinated, $X$ indicates whether they are Republican, and $X_N$ is the proportion of Republicans in their county. Individuals all have equal weight and are matched in a tiered approach, where Tier 1 contains the highest confidence matches and Tier 6 contains the lowest confidence matches. Individuals with multiple matches are discarded.}
\label{tab:zero_tie_unitwt}
\end{longtable}

\renewcommand{\arraystretch}{2.5}
\begin{longtable}{|c|c|c|c|c|c|c|}
    \hline
    & \textbf{Tier 1} & \textbf{Tier 1-2} & \textbf{Tier 1-3} & \textbf{Tier 1-4} & \textbf{Tier 1-5} & \textbf{Tier 1-6} \\
    \hline
    \shortstack{$E[Cov(Y, X_N|X)]$} & 
    \shortstack{-0.00463\\(8e-05)} & 
    \shortstack{-0.00451\\(8e-05)} & 
    \shortstack{-0.0045\\(7e-05)} & 
    \shortstack{-0.00446\\(7e-05)}&
    \shortstack{-0.00459\\(7e-05)}& 
    \shortstack{-0.00481\\(7e-05)} \\
    \hline
    \shortstack{$E[Cov(Y, X|N)]$} & 
    \shortstack{-0.00698\\(0.00022)} & 
    \shortstack{-0.00706\\(0.00021)} & 
    \shortstack{-0.00638\\(0.0002)} & 
    \shortstack{-0.00633\\(0.0002)}&
    \shortstack{-0.00594\\(0.00019)}& 
    \shortstack{-0.00519\\(0.00019)} \\
    \hline
\caption{Conditional covariance (standard error) estimated from the auxiliary individual-level dataset. $Y$ indicates whether an individual is vaccinated, $X$ indicates whether they are Republican, and $X_N$ is the proportion of Republicans in their county. Individuals all have equal weight and are matched in a tiered approach, where Tier 1 contains the highest confidence matches and Tier 6 contains the lowest confidence matches. For individuals with multiple matches, one match is chosen at random.}
\label{tab:random_tie_unitwt}
\end{longtable}

\renewcommand{\arraystretch}{2.5}
\begin{longtable}{|c|c|c|c|c|c|c|}
    \hline
    & \textbf{Tier 1} & \textbf{Tier 1-2} & \textbf{Tier 1-3} & \textbf{Tier 1-4} & \textbf{Tier 1-5} & \textbf{Tier 1-6} \\
    \hline
    \shortstack{$E[Cov(Y, X_N|X)]$} & 
    \shortstack{-0.00361 \\ (0.00042)} & 
    \shortstack{-0.00363\\(0.00041)} & 
    \shortstack{-0.00353\\(0.00038)} & 
    \shortstack{-0.00354\\(0.00037)}&
    \shortstack{-0.0036\\(0.00037)}& 
    \shortstack{-0.00366\\(0.00036)} \\
    \hline
    \shortstack{$E[Cov(Y, X|N)]$} & 
    \shortstack{-0.00847\\(0.001)} & 
    \shortstack{-0.00876\\(0.00096)} & 
    \shortstack{-0.00833\\(0.00089)} & 
    \shortstack{-0.00712\\(0.0009)}&
    \shortstack{-0.00691\\(0.00088)}& 
    \shortstack{-0.00656\\(0.00087)} \\
    \hline
\caption{Conditional covariance (standard error) estimated from the auxiliary individual-level dataset. $Y$ indicates whether an individual is vaccinated, $X$ indicates whether they are Republican, and $X_N$ is the proportion of Republicans in their county. Individuals are re-weighted (see Appendix \ref{sec:app_application}) and are matched in a tiered approach, where Tier 1 contains the highest confidence matches and Tier 6 contains the lowest confidence matches. Individuals with multiple matches are discarded.}
\label{tab:zero_tie_rewt}
\end{longtable}

\newpage
\section{Auxiliary Matched Analysis: Additional Detail}
\label{sec:app_application}

\subsection*{Data}
\subsubsection*{Individual COVID-19 Vaccinations}
We used data on individual COVID-19 vaccination status from the American Family Cohort (AFC). The AFC is derived from the American Board of Family Medicine PRIME Registry \citep{balraj2023american}. PRIME is used by practices to track electronic health records data for Medicare and Medicaid reporting as well as measurement of social determinants of health \citep{phillips2017prime}. PRIME contains approximately 8 million patient records across 47 states \citep{balraj2023american}. From this dataset, we obtained COVID-19 vaccination status following Hao et al., 2023 \citep{hao2023covid}. We considered all patients over 5 years old who did not die during 2021. Patients who had no history of receiving immunizations or visiting practices operating from January 1, 2019 to December 31, 2021 were excluded. We considered patients to be vaccinated if they received two COVID-19 vaccinations, or a single dose of the Janssen vaccine, before December 18, 2021 (n = 339,484). All other patients (n = 1,463,303) were considered unvaccinated.

\subsubsection*{Individual Voter Information}
We used national voter records from the L2 database \citep{l2data2024}. This database is collected and assembled by the L2 organization and sold for use in academic, political, or business settings. In addition to minimizing data gathering steps, the database was constructed using a number of standardization and cleaning preprocessing steps that ensure greater comparability across states. We obtained voter name, state, date of birth, and registered political party for 214,125,844 voters nationally.

\subsubsection*{County-Level Voter Information}
We used election returns at the county level from the 2020 Presidential Election from the \citep{DVN/VOQCHQ_2020} to provide the $b_i$ in the individual-level analysis to check the signs of the covariance terms.

\subsection*{Matching}
We used a tiered matching strategy to pair patients from the AFC dataset with voters from the L2 dataset. This strategy was informed by \citep{roos1986art}, who implement a tiered matching system for records with few identifiers. We implement a similar system, where ``Tier 1" matches represent high confidence, and ``Tier 6" matches represent low confidence (see Table \ref{tab:patient_tiers_0tie}, \ref{tab:patient_tiers_randomtie}). 

We used the characteristics of the datasets to determine which factors to match on in each tier. In the L2 data, we found that recorded birth dates were non-uniformly distributed, with strong biases towards birth dates on the first of each month and on January 1. This suggests that the L2 data may record births as occurring on the first of the month or on January first when respondents only note their birth month or year. Thus, in some tiers, we allow for matches without identical birth dates as long as the birth month and year are the same. 

We require an exact match on first and last name in all tiers. We considered matching on first initial and last name, but we found that these matching schemes resulted in clearly erroneous matches. We also considered matching on phonetic last name. However, since our datasets come from official voter and patient records, it is unlikely that last names are misspelled. Thus, a phonetic encoding may have led to incorrect matches. 

In our final matching scheme, tier 1 requires matches at the highest level of specificity with exact matches across all available factors. We are most confident that tier 1 matches are correct, since they match for all six factors. Tier 6 only requires exact matches on three factors, so we are less confident that these matches are correct. The tiers are defined as follows:

\begin{itemize}
    \item Tier 1: Unique match on exact first name, last name, state, and birth date (day/month/year)
    \item Tier 2: Unique match on exact first name, last name, state, and birth month/year
    \item Tier 3: Unique match on exact first name, last name, state, and birth year
    \item Tier 4: Unique match on exact first name, last name, and birth date (day/month/year)
    \item Tier 5: Unique match on exact first name, last name, and birth month/year
    \item Tier 6: Unique match on exact first name, last name, and birth year
\end{itemize}

We tested two different ways of considering ties, where multiple vaccination records or voters were matched to one voter or vaccination record in a given tier. In the first approach, we did not allow any ties (``zero-tie"), deleting ties from matched data at each tier. In the second approach, following \citep{roos1986art}, we kept one record at random from each set of ties (``random-tie"). We present analyses of the zero-tie strategy in the main text. Analyses for the random-tie strategy are presented in the supplement.

\begin{longtable}{|c|c|c|c|}
\hline
\textbf{Tier} & \textbf{Uniquely Matched} & \textbf{Unmatched} & \textbf{Proportion Matched} \\
\hline
Tier 1 & 1,017,300 & 764,810 & 0.563 \\
Tier 2 & 57,517 & 702,882 & 0.595 \\
Tier 3 & 101,999 & 573,883 & 0.651 \\
Tier 4 & 39,042 & 563,416 & 0.673 \\
Tier 5 & 26,207 & 511,227 & 0.687 \\
Tier 6 & 39,464 & 414,351 & 0.709 \\
\hline
\caption{Summary of matched patient counts for the auxiliary individual-level dataset. Tier 1 contains the highest confidence matches and Tier 6 contains the lowest confidence matches. Individuals with multiple matches are discarded.}
\label{tab:patient_tiers_0tie}
\end{longtable}

When estimating covariances and confidence intervals, we tested two different weighting schemes. In the first scheme, each record was given equal weight (``unit weight"). In the second scheme, we re-weighted records such that the sum of records for a given county in a given political party matched the number of voters of that party in that county, as calculated from our county-level voter information (``re-weighted"). Because many counties had few records in our matched dataset, we restricted the re-weighted data to only weight patients coming from counties with over 20 records -- all other patients were given weight 0. We present analyses for the unit weight strategy in the main text. Analyses for the re-weighted strategy are presented in the supplement.

When patient counties recorded in the AFC and L2 data did not match, we chose to use the AFC county. For the zero-tie and random-tie approaches, 4,346 and 5,235 records respectively were removed due to missing AFC county. Additionally, in Alaska, the county-level dataset recorded Alaskan districts rather than counties. We removed all people with AFC counties in Alaska from the merged dataset. We also restricted to only Democrat and Republican voters. This left us with a total of 987,663 (437,722 Democrat, 549,941 Republican) and 1,061,517 (476,496 Democrat, 585,021 Republican) individuals matched in the zero-tie and random-tie approaches, respectively. 

\newpage
\subsection*{Estimating Conditional Covariances}
In this section, we describe how we estimate $\theta = \E[\Cov(A,B|C)]$, where $C$ takes values in $\mathcal{C}$, and $\Var(\hat{\theta})$. We assume that the distribution of $C$ is known, as is the distribution of $A|C$.

We know:
\begin{align*}
    \theta &= \sum_{c\in \mathcal{C}} \Pr(C=c) \Cov(A,B|C=c)\\
\end{align*}

which we estimate by
\begin{align*}
    \hat{\theta} &= \sum_{c\in \mathcal{C}} \Pr(C=c) \sum_{i: C_i = c} \frac{1}{n_c-1} (A_i - \bar{A})(B_i-\bar{B}) \\
    &= \sum_{c\in \mathcal{C}} \frac{\Pr(C=c)}{n_c-1} \sum_{i: C_i = c} (A_i - \bar{A})(B_i-\bar{B}) \\
    &= \sum_{c\in \mathcal{C}} \Pr(C=c) \left( \frac{\sum_{i: C_i = c}(A_i-\bar{A})^2}{n_c-1}  \right) \left(\frac{\sum_{i: C_i = c}(A_i - \bar{A})(B_i-\bar{B})}{\sum_{i: C_i = c}{(A_i-\bar{A})^2}}\right)\\
    &= \sum_{c\in \mathcal{C}} \Pr(C=c) \Var(A|C=c) \hat{\beta}_c
\end{align*}

where $\hat{\beta_c}$ is the coefficient from the linear regression of $A$ on $B$ when $C=c$, $n_c$ is the number of entries where $C=c$, and $\Var(A|C=c)$ is the variance of $A$ for entries where $C=c$. 

We now estimate $\Var(\hat{\theta})$.

\begin{align*}
    \Var(\hat{\theta}) &= \Var(\sum_{c\in \mathcal{C}} \Pr(C=c) \Var(A|C=c) \hat{\beta_c})\\
    &= \sum_{c\in \mathcal{C}} \Var(\Pr(C=c) \Var(A|C=c) \hat{\beta_c})\\
    &= \sum_{c\in \mathcal{C}} \Pr(C=c)^2 \Var(A|C=c)^2 \Var(\hat{\beta_c})\\
    &= \sum_{c\in \mathcal{C}} \Pr(C=c)^2 \Var(A|C=c)^2 \text{SE}({\hat{\beta_c}})^2
\end{align*}

Where $\text{SE}({\beta_c})^2$ is estimated standard error from the linear regression of $A$ on $B$ when $C=c$.

In our setting, we seek to estimate $\E\left[\Cov(X_N,Y|X)\right]$ and $\E\left[\Cov(X,Y|X_N)\right]$. We estimate $\E\left[\Cov(X_N,Y|X)\right]$ as follows:

We calculate $Pr(X = 1)$ as the weighted sum of entries where $X=1$ over the total weight of the dataset. We then calculate $Var(X_N|X=i)$. For each value of $X$, we then perform weighted least squares regression ($Y \sim X_N|X=i$). We calculate the coefficient on $X_N$ of this regression, $\hat{\beta_{i}}$. For our estimate for the standard error of $\E\left[\Cov(X_N,Y|X)\right]$, we calculate the standard error of $\hat{\beta_{i}}$ and follow the formulas above.

For $\E\left[\Cov(X,Y|X_N)\right]$, we first note a useful proposition:

\begin{prop}\label{prop:equivalence}
    We can write that:
    \begin{align*}
        \E[\Cov(X,Y|N)] = \E[\Cov(X,Y|X_N)]
    \end{align*}
\end{prop} 

Using this proposition, we estimate $\E\left[\Cov(X,Y|N)\right] = \E\left[\Cov(X,Y|X_N)\right]$ as follows:

We round $X_N$ to the nearest $.05$, creating 19 groups $(d_1,...d_{19})$ of $X_N$ values (in this case, $X_N$ = county proportion Republican). We estimate $Pr(X_N \in d_i)$ as the weighted sum of entries with $X_N\in d_i$ over the total weight of the dataset. We also calculate the variance of $X$, $\Var(X|X_N \in d_i)$. We then perform weighted least squares regression ($Y \sim X|X_N \in d_i$) and calculate the coefficient on $X$ of this regression, $\widehat{\beta_{d_i}}$. For our estimate for the standard error of $\E\left[\Cov(X,Y|X_N)\right]$, we calculate the standard error of $\widehat{\beta_{d_i}}$, $\text{SE}(\widehat{\beta_{d_i}})$ and follow the formulas above.

\begin{proof}[Proof of \textbf{Proposition \ref{prop:equivalence}}]
    By the law of total covariance, the left side is:
    \begin{align*}
    \mathbb{E}[\Cov(X,Y|N)] =\Cov(X,Y) -\Cov(\mathbb{E}[X|N],\mathbb{E}[Y|N])
    \end{align*}
    and similarly the right side is :
    \begin{align*}
        \mathbb{E}[\Cov(X,Y|X_N)] = \Cov(X,Y) -\Cov(\mathbb{E}[X|X_N],\mathbb{E}[Y|X_N])
    \end{align*}
    So it is enough that:
    \begin{align*}
        \mathbb{E}[\mathbb{E}[X|N]\mathbb{E}[Y|N]] -\mathbb{E}[\mathbb{E}[X|N]]\mathbb{E}[\mathbb{E}[Y|N]] = \mathbb{E}[\mathbb{E}[X|X_N]\mathbb{E}[Y|X_N]] - \mathbb{E}[\mathbb{E}[X|X_N]]\mathbb{E}[\mathbb{E}[Y|X_N]].
    \end{align*}
    But notice that the second term on both sides is $\mathbb{E}[X]\mathbb{E}[Y]$ by the law of iterated expectations, while $\mathbb{E}[X|X_N]=\mathbb{E}[X|N]=X_N$, by definition of $X_N$. Thus, we simply require that:
    \begin{align*}
        \mathbb{E}[X_N\cdot Y_N] = \mathbb{E}[X_N \cdot \mathbb{E}[Y|X_N]].
    \end{align*}
But
    \begin{align*}
\mathbb{E}[X_N \cdot Y_N] = \mathbb{E}[\mathbb{E}[X_N\cdot Y_N|X_N]] = 
\end{align*}
by the Law of Iterated Expectations, and $X_N$ is constant given $X_N$, so 
\begin{align*}
    \mathbb{E}[X_N\cdot Y_N] = \mathbb{E}[X_N \cdot \mathbb{E}[Y_N|X_N]].
\end{align*}
Now since $Y_N$ is $\mathbb{E}[Y|N]$, we can write that:
\begin{align*}
    \mathbb{E}[X_N \cdot \mathbb{E}[Y_N|X_N]] = \mathbb{E}[X_N\cdot \mathbb{E}[\mathbb{E}[Y|N]|X_N]].
\end{align*}

But note that:
\begin{align*}
    \mathbb{E}[\mathbb{E}[Y|N]|X_N] = \mathbb{E}[Y|X_N]
\end{align*}
by the (general version of the) Law of Iterated Expectations. Hence,
\begin{align*}
    \mathbb{E}[X_N\cdot Y_N] = \mathbb{E}[X_N\cdot \mathbb{E}[Y|X_N]]
\end{align*}
as desired. 

\end{proof}

\end{document}